\documentclass[12pt]{article}
\usepackage[margin=1in]{geometry}

\usepackage[utf8]{inputenc}
\usepackage{authblk}
\usepackage{subcaption}

\title{Online Search With Best-Price and \\ Query-Based Predictions}

\author{Spyros Angelopoulos$^1$, Shahin Kamali$^{2}$, and Dehou Zhang$^2$}
\date{{\small $^1$ CNRS and Sorbonne University, Paris, France. \\ $^2$ University of Manitoba, Winnipeg, Canada. }}

\usepackage{times}  
\usepackage{helvet}  
\usepackage{courier}  
\usepackage[hyphens]{url}  
\usepackage{graphicx} 
\usepackage{caption} 
\usepackage{algorithm}
\usepackage{algpseudocode}

\usepackage[T1]{fontenc}

\usepackage{newfloat}
\usepackage{listings}
\floatstyle{ruled}
\newfloat{listing}{tb}{lst}{}
\floatname{listing}{Listing}
\usepackage{amssymb,amsmath,amsthm}
\usepackage{mathtools} 
\usepackage{xspace}
\newtheorem{theorem}{Theorem}
\newtheorem{lemma}{Lemma}
\newtheorem{corollary}{Corollary}
\usepackage{tikz}

\newcommand{\comp}{\textsc{cr}}
\newcommand{\robust}{{\sc Robust-Mix}\xspace}
\newcommand{\ORA}{{\sc Ora}{{\ensuremath{_r}}}\xspace}
\newcommand{\RLIS}{{\sc Rlis}\xspace}
\newcommand{\RBIS}{{\sc Rbis}\xspace}
\newcommand{\ONstar}{{\sc On}{{\ensuremath{^*}}}\xspace}

\newcommand{\crchart}{
	\centering
	\includegraphics[width = .7\columnwidth, clip, trim = 0mm .8cm 0mm 0mm]{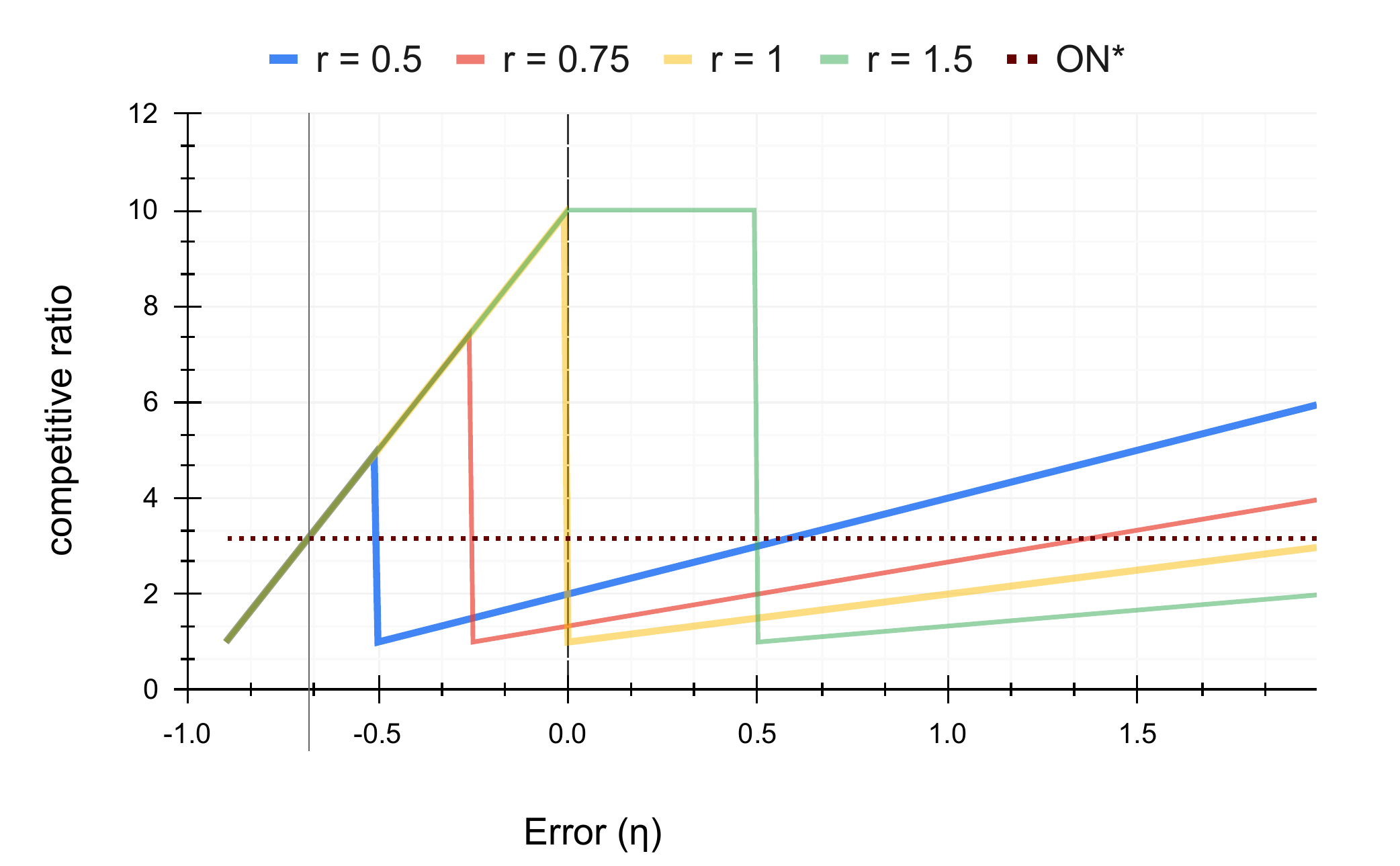}
	
	{\tiny \vspace*{-.47cm}\hspace*{1.4cm}$\sqrt{m/M}$} \\ \scriptsize{ \hspace*{1.99cm} {\fontfamily{bch}\selectfont negative $\leftarrow \eta \rightarrow$ positive} }}

\newcommand{\ourtitle}{Online Search with Best-Price and Query-Based Predictions}

\begin{document}

\maketitle

\begin{abstract}
In the online (time-series) search problem, a player is presented with a sequence of prices which are revealed in an online manner. In the standard definition of the problem, for each revealed price, the player must decide irrevocably whether to accept or reject it, without knowledge of future prices (other than an upper and a lower bound on their extreme values), and the objective is to minimize the competitive ratio, namely the worst case ratio between the maximum price in the sequence and the one selected by the player. The problem formulates several applications of decision-making in the face of uncertainty on the revealed samples.

Previous work on this problem has largely assumed extreme scenarios in which either the player has almost no information about the input, or the player is provided with some powerful, and error-free advice. In this work, we study learning-augmented algorithms, in which there is a potentially erroneous prediction concerning the input. Specifically, we consider two different settings: the setting in which the prediction is related to the maximum price in the sequence, as well as the setting in which the prediction is obtained as a response to a number of binary queries. For both settings, we provide tight, or near-tight upper and lower bounds on the worst-case performance of search algorithms as a function of the prediction error. We also provide experimental results on data obtained from stock exchange markets that confirm the theoretical analysis, and explain how our techniques can be applicable to other learning-augmented applications.  
\end{abstract}

\ \newpage
\section{Introduction}
\label{sec:introduction}

The {\em online (time series) search} problem formulates a fundamental setting in decision-making under uncertainty. In this problem, a player has an indivisible asset that wishes to sell within a certain time horizon, e.g., within the next $d$ days, without knowledge of $d$. On each day $i$, a {\em price} $p_i$ is revealed, and the player has two choices: either accept the price, and accrue a {\em profit} equal to $p_i$, or reject the price, in which case the game repeats on day $i+1$. If the player has not sold by day $d$ (i.e., has rejected all prices  $p_1, \ldots ,p_{d-1}$), then the last price $p_d$ is accepted by default. 

This problem was introduced and studied in~\cite{el2001optimal} by means of {\em competitive analysis}. Namely, the competitive ratio of the player's strategy (or algorithm) is defined as the worst case ratio, over all price sequences, of the maximum price in the sequence divided by the price accepted by the player. Thus, the competitive ratio provides a worst-case guarantee that applies even to price sequences that are adversarially generated. Since the problem formulates a basic, yet fundamental transaction setting, a player that follows a competitively efficient algorithm has a safeguard against any amount of volatility with respect to prices.

El-Yaniv et al. \cite{el2001optimal} gave a simple, deterministic algorithm that achieves a competitive ratio equal to $\sqrt{M/m}$, where $M,m$ are upper and lower bounds on the maximum and minimum price in the sequence, respectively, and which are assumed to be known to the algorithm. This bound is tight for all deterministic algorithms. Randomization can improve the competitive ratio to an asymptotically tight bound equal to $O(\log (M/m))$. See also the surveys~\cite{el1998competitive, mohr2014online}. 

Online search is a basic paradigm in the class of online financial optimization problems. Several variants and settings have been studied through the prism of competitive analysis; see, e.g.,~\cite{damaschke2009online, lorenz2009optimal, xu2011optimal, clemente2016advice}. The problem has also been studied as a case study for evaluating several performance measures of online algorithms, including measures alternative to competitive analysis~\cite{boyar2014comparison,ahmad2021analysis}. Extensions of online search such as {\em one-way trading} and {\em portfolio selection} have also been studied extensively both within competitive analysis;
e.g.,~\cite{el2001optimal, fujiwara2011average, borodin2000competitive}, as well as from the point of view of regret minimization; e.g.,~\cite{hazan2015online, uziel2020long, das2014online}. We refer also to the survey~\cite{li2014online}.


Previous work on competitive analysis of online financial optimization problems, including online search, has largely assumed a status of almost complete uncertainty in regards to the input. Namely, the algorithm has either no knowledge, or very limited knowledge concerning the input.
 This models a scenario that is overly pessimistic: indeed, in everyday financial transactions, the players have some limited, albeit potentially erroneous information on the market. 

This observation illustrates the need for a study of online financial optimization problems using the framework of {\em learning-enhanced} competitive algorithms~\cite{DBLP:conf/icml/LykourisV18,NIPS2018_8174}. Such algorithms have access to some machine-learned information on the input which
is associated with a {\em prediction error} $\eta$. The objective is to design algorithms whose competitive ratio degrades gently as the prediction error increases, but also quantify the precise tradeoff between the performance and the prediction error. Several online optimization problems have been studied in this setting, including caching~\cite{DBLP:conf/icml/LykourisV18,rohatgi2020near}, ski rental and non-clairvoyant scheduling~\cite{NIPS2018_8174,WeiZ20}, makespan scheduling~\cite{lattanzi2020online}, 
rent-or-buy problems~\cite{DBLP:conf/nips/Banerjee20,anand2020customizing,gollapudi2019online}, 
secretary and matching problems~\cite{DBLP:conf/nips/AntoniadisGKK20,abs-2011-11743}, and metrical task systems~\cite{DBLP:conf/icml/AntoniadisCE0S20}. See also the survey~\cite{mitzenmacher2020algorithms}. 
A related line of research is the untrusted advice framework proposed by~\cite{DBLP:conf/innovations/0001DJKR20} in which the algorithm’s performance is evaluated at the extreme cases in which the advice is either error-free, or adversarially generated. 

To our knowledge, there is no previous work on competitive analysis of 
online financial optimization problems in the learning-enhanced model. Note that this is in contrast to analysis based on regret minimization, which inherently incorporates predictions as ``experts''~\cite{cover1996universal,hazan2007online}.  


\subsection{Contribution}
\label{subsec:contribution}

We present the first results on competitive online search in a setting that provides predictions related to the price sequence. We show that the obtained competitive ratios are optimal under several models. We also introduce new techniques for leveraging predictions that we argue can be applicable to other learning-augmented online problems. More precisely, we study the following two settings:

\paragraph{The prediction is the best price} Here, the prediction is the best price that the player is expected to encounter. We further distinguish between the model in which no other information on this prediction is available to the player, which we call the {\em oblivious} model, and the model in which an {\em upper bound} to the prediction error is known, which we call the {\em non-oblivious} model. In the latter, the player knows that the error is bounded by some given value $H$, i.e., $\eta \leq H$. The oblivious model is more suitable for markets with very high volatility (e.g., cryptocurrencies), whereas the non-oblivious model captures less volatile markets (e.g., fiat currencies), in which we do not expect the prices to fluctuate beyond a (reasonable) margin. For both models, we give optimal (tight) upper and lower bounds on the competitive ratio as function of the prediction error. A novelty in the analysis, in comparison to previous work, is that we perform an asymmetric analysis in regards to the error, namely we distinguish between {\em positive} and {\em negative} error, depending on whether the best price exceeds the prediction or not. This distinction is essential in order to prove the optimality of our results. 

\paragraph{The prediction is given as response to binary queries} In this model, the prediction is given as a response to $n$ binary queries, for some fixed $n$. For example, each query can be of the form  ``will a price at least equal to $p$ appear in the sequence?''. This model captures settings in which the predictions define {\em ranges} of prices, as opposed to the single-value prediction model, and was introduced recently in the context of a well-known resource allocation problem in AI, namely the {\em contract scheduling} problem~\cite{DBLP:conf/aaai/0001K21}. The prediction {\em error} is defined as the number of erroneous responses to the queries, and we assume non-oblivious algorithms which know an upper bound $H<n$ on the error. Online search was previously studied under an error-free query model in~\cite{clemente2016advice}, however their proposed solution is non-robust: a single query error can force the algorithm to accept a price as bad as the smallest price in the sequence.

We present two different algorithms in this model, and prove strict upper bounds on their competitive ratios, as functions of $n$ and $H$. The first algorithm uses the $n$ queries so as to choose a price from $n$ suitably defined intervals, then accepts the first price in the sequence that is at least as high as the chosen price; moreover, its performance is guaranteed as long as at most half of the query responses are correct. We then present an algorithm based on  {\em robust binary search}, which allows to select a suitable price from a much larger space of $2^n$ intervals, thus leading to improved performance, at the expense of a relatively smaller (but still high) tolerance to errors (i.e., the theoretical analysis assumes that $H<n/4$). This result is the main technical contribution of this work, and we expect that it can find applications in many other settings in which we must identify a ``winner'' from a large set of candidates, in the presence of errors. We give such a concrete application in Section~\ref{sec:conclusion}. We complement the robust binary-search upper bound with a theoretical lower bound on the competitive ratio in this query model.

For both models, we evaluate experimentally our algorithms on real-world data, in which the prices are the exchange rates for cryptocurrencies or fiat currencies. Our experimental results demonstrate that the algorithms can benefit significantly from the predictions, in both models, and that their performance decreases gently as the error increases. 

\subsection{Notation and definitions}
\label{subsec:notation}

Let $\sigma=(\sigma_i)_{i=1}^d$ be a sequence of prices revealed on days $1,d$. 
Given an algorithm $A$, $A(\sigma)$ is the {\em profit} of $A$ on $\sigma$; 
since $\sigma$ and  $A$ are often implied from context, we simply refer to the profit of the algorithm as its accepted price.  We denote by $p^*$ the optimal price in the input. Given an algorithm $A$, we denote its competitive ratio by \comp(A). Recall that $m,M$ are the known lower and upper bounds on the prices in the sequence, respectively. We denote by \ONstar the optimal online algorithm without predictions, i.e., the algorithm of competitive ratio $\sqrt{M/m}$.

A {\em reservation} algorithm with price $q$ is an algorithm that accepts the first price in the sequence that is at least $q$. For example, it is known that \ONstar can be described as a reservation algorithm with price $\sqrt{M/m}$.

\section{Algorithms with Best-Price Prediction}


In this setting, the prediction is the highest price that will appear in the sequence. In the remainder of the section, we denote this prediction with $p \in [m,M]$, and recall that $p^*$ is the optimal price. The prediction $p$ is associated with an error $\eta$, defined as follows. 
If $p^* \leq p$, we define $\eta$ to be such that $1-\eta = p^*/p$, and we call this error {\em negative}, in the sense that the best price is no larger than the predicted price.  Note also that the negative error ranges in $[0,(M-m)/M]$, that is, $\eta<1$, in this case. 
If $p^*>p$, we define $\eta$ to be such that
$1+\eta=p^*/p$, and we call the error {\em positive}, in the sense that the best price is larger than the predicted price. Since $1 < p^*/p \leq M/m$, the positive error ranges in $(0,(M-m)/m]$. 
Naturally, the online algorithm does not know neither the error value, nor its parity. The parity is a concept that we introduce for the benefit of the analysis. 

Depending on the volatility of the market, the positive and negative error can fluctuate within a certain range. Let $H_n$, $H_p$ denote upper bounds on the negative and positive errors, respectively, i.e., $H_n \leq (M-m)/M$, and $H_p \leq (M-m)/m$. We distinguish between non-oblivious and oblivious algorithms, namely between algorithms that know $H_n$ and $H_p$, and algorithms that do not, respectively.

\subsection{Oblivious algorithms}
\label{subsec:price.oblivious}

We first study oblivious algorithms, and show matching upper and lower bounds on the competitive ratio. Given algorithm $A$ with prediction $p$, define the function $s_A(p,m,M) \in [m,M]$ as the smallest price revealed on day 1 (i.e., the smallest value of $p_1$) such that $A$ accepts that price on day 1. 
Define also $r_A=s_A(p,m,M)/p$. We first show a lower bound on the competitive ratio. 

\begin{theorem}\label{th:lowerBound}
	For any algorithm $A$ with prediction $p$, 
	\[ \comp(A) \geq  
	\begin{cases}
		(1-\eta)/r_A, & \text{if } \eta \leq 1-r_A \\
		{(1-\eta)M}/{m}, & \text{if } \eta > 1-r_A,
	\end{cases}
	\]  
	if the error is negative, and  
	\[
	\comp(A) \geq  
	\begin{cases}
		(1+\eta)/r_A, & \text{if } \eta \geq r_A-1 \\
		{M}/{m}, & \text{if } \eta < r_A-1,\\
	\end{cases} 
	\]  
	if the error is positive.
\end{theorem}

\begin{proof}
	\noindent
	{\em Case 1: $\eta$ is negative}, i.e., $p^* = p(1-\eta)$. The adversary chooses $p^*=(1-\eta)M$, which implies that $p=M$. 
	
	Suppose first that $r_A \leq 1-\eta$, hence $r_A\cdot p \leq p^*$. The adversary presents the sequence
	$r_A \cdot p, p^*, \ldots ,p^*$ . From the definition of $r_A$,  $A$ accepts the price on day 1, and $\comp(A) \geq p^*/(r_A\cdot p) 
	= (1-\eta)/r_A$.
	
	Next, suppose that $r_A > 1-\eta$. We have $r_A\cdot p  > (1-\eta) \cdot p =  p^*$. 
	The adversary presents the sequence  $p^*, m, \ldots ,m$. By definition, $A$ rejects the price on day 1, hence its profit is $m$, and $\comp(A) \geq p(1-\eta)/m = (1-\eta) M/m.$
	
\smallskip

	\noindent 
	{\em Case 2: $\eta$ is positive}, i.e., $p^* = p(1+\eta)$. 
	The adversary chooses $p^* = M$, which implies that $p=M/(1+\eta)$.
	
	Suppose first that $r_A \leq 1+\eta$, that is, $r_A \cdot p \leq p^*$. The adversary chooses the sequence of prices $r_A \cdot p, p^*, p^*, \ldots ,p^*$. By definition, $A$ accepts on day 1, therefore
	$\comp(A) \geq p^*/(r_A\cdot p) = (1+\eta)/r_A$.
	
	Next, suppose that $r_A > 1+\eta$, which implies $r_A\cdot p  > (1+\eta) \cdot p =  p^* =M$. The intuition here is that $A$ does not accept on day 1 a price equal to $M$, which is clearly a bad decision. The adversary chooses the sequence of prices $M, m, \ldots ,m$, and thus $\comp(A) \geq M/m$.
\end{proof}

Next, we show a class of algorithms whose competitive ratio matches Theorem~\ref{th:lowerBound}. 
For any $r>0$, define the {\em oblivious reservation algorithm}, named \ORA as the algorithm with reservation price $r\cdot p$, given the prediction $p$. 

\begin{theorem}\label{lem:up}
	The  algorithm \ORA (with reservation price $r \cdot p$) has competitive ratio 
\[ \comp({\textsc{Ora}}_r) \leq  
	\begin{cases}
		(1-\eta)/r, & \text{if } \eta \leq 1-r \\
		{(1-\eta)M}/{m}, & \text{if } \eta > 1-r,
	\end{cases}
	\]  
	if the error is negative, and  
	\[
\comp({\textsc{Ora}}_r) \leq  
	\begin{cases}
		(1+\eta)/r, & \text{if } \eta \geq r-1 \\
		{M}/{m}, & \text{if } \eta < r-1,\\
	\end{cases} 
	\]  
	if the error is positive.
\end{theorem}

\newcommand{\compORA}{\textsc{cr(Ora$_r$)}}
\begin{proof} \ \\ 
	\noindent \textbf{negative error:} Let $\eta \leq 1$ be a negative error. Thus, we have $p^* = p(1-\eta)$. 
	
	\noindent First, suppose $r \leq (1-\eta)$, which means $r\cdot p \leq p^*$. Therefore, \ORA has a profit of at least $r \cdot p$, and thus $\compORA \leq \frac{(1-\eta)p}{r \cdot p} = \frac{1-\eta}{r} $.
	
	\noindent Next, suppose $r > (1-\eta)$. We have $p^* = p (1-\eta) \leq (1-\eta) \cdot M $. On the other hand, the profit of \ORA is at least $m$. Therefore, we have $\compORA \leq (1-\eta)M/m$. 
	
	\noindent \textbf{positive error:} Let $\eta$ be a positive error. Thus, we have $p^* = p(1+\eta)$.
	
	\noindent First, suppose $r \leq (1+\eta)$, that is, $r\cdot p \leq p^*$, and \ORA has a profit of at least $r \cdot p$. On the other hand, we have $p^* = p(1+\eta)$, and thus $\compORA \leq \frac{p(1+\eta)}{r \cdot p} = \frac{1+\eta}{r} $.
	
	\noindent Next, suppose $r > (1+\eta)$. Then, we have $p^* \leq M$ and the profit of \ORA is at least $m$, and thus $\compORA \leq M/m$. 
\end{proof}



Figure~\ref{fig:crChart} illustrates the competitive ratio of \ORA, as function of $\eta$, for different values of the parameter $r$. 
First, we observe that there is no value of $r^*$ such that {\textsc{Ora}$_{r^*}$} dominates \ORA with $r \neq r^*$. More precisely, for any pair of $r_1$ and $r_2$, there are some values of $\eta$ for which {\textsc{Ora}$_{r_1}$}  has a better competitive ratio while for other values of $\eta$, {\textsc{Ora}$_{r_2}$} has a better competitive ratio. 

For positive error, the competitive ratio degrades linearly in the error with slope $1/r$. Note however that if $r=1.5$, we have $\compORA=M/m (=10)$ for 
$\eta < r - 1 = 0.5$, since for positive error in $(0,0.5)$, \ORA does not trade at day even if the trading price is $M$. 
For the other values of $r$, we have $r \leq \eta+1$ for which $\compORA\leq (1+\eta)/r$, and hence the algorithm performs better when $r$ becomes larger. 

For negative but {\em small} values of error, we have that $\compORA=(1-\eta)/r$, 
thus the performance improves linearly in the error, again with slope $1/r$. For larger values, i.e., if $\eta>1-r$, there is a ``jump'' in the competitive ratio, which increases from $1$ to $(1-r) M/m$. Following this jump, $\compORA$ improves linearly with the error, this time with slope $M/m$.

\begin{figure}
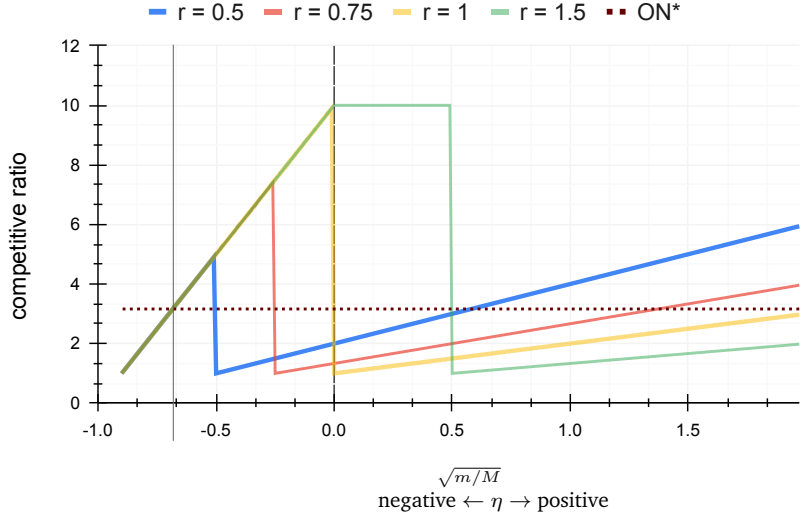

	\crchart
	\caption{The competitive ratio of \ORA as function of the error $\eta$, and the parameter $r$. Here we choose $M/m=10$, and thus the ranges of negative and positive error are $[0,0.9]$ and $(0,9)$, respectively. } 
	\label{fig:crChart}
\end{figure}

Even though \ORA is optimal according to Theorems~\ref{th:lowerBound} and~\ref{lem:up}, its competitive ratio may be worse than \ONstar for certain ranges of error (namely, for small negative or large positive error). However, this is unavoidable: the next corollary shows that there is no oblivious algorithm with best-price prediction that improves upon \ONstar for all values of the error, i.e., cannot dominate \ONstar for all values of error. 

\begin{corollary}\label{coro:coro}
	For any oblivious algorithm $A$ with best-price prediction, there exists some range of error $\eta$ for which 
	$\comp(A) > \sqrt{M/m}$. 
\end{corollary}

\begin{proof}
	First, suppose that the error is negative and in the range $(1-r_A, 1-\sqrt{m/M})$. Theorem~\ref{th:lowerBound} shows that  $\comp(A) > (1-\eta)M/m\geq (\sqrt{m/M}) M/m = \sqrt{M/m}$.
	Next, suppose that the error is positive and $\eta \geq r_A \sqrt{M/m} -1$. Therefore, $\eta > r_A-1$ and again by Theorem~\ref{th:lowerBound}, the $\comp(A) \geq (1+\eta)/r_A \geq \sqrt{M/m}$.
\end{proof}

\subsection{Non oblivious algorithms}
\label{subsec:aware}

In this section, we show matching upper and lower bounds on the competitive ratio, in the setting 
in which the algorithm knows upper bounds $H_n$ and $H_p$ on the negative and positive error, respectively. 

We call an algorithm A \emph{robust} if for all values of $M$ and $m$, and all values of $\eta$ (negative or positive), $\comp(A)\leq \sqrt{M/m}$. In light of Corollary~\ref{coro:coro}, without knowing $H_n$ and $H_p$, no online algorithm can be robust. 
In what follows, we will show that there exist robust non-oblivious algorithms. In particular, we define the algorithm \robust which works as follows. If $(1+H_p)/(1-H_n) > \sqrt{M/m}$, then \robust ignores the prediction and applies \ONstar. Otherwise, i.e., if $(1+H_p)/(1-H_n) \leq \sqrt{M/m}$, \robust is an algorithm with reservation price equal to  $p' = p(1-H_n)$. 
Note that since $H_n \leq (M-m)/M$, we have $p' \geq mp/M$.


\begin{theorem} \label{th:haware} 
	$\comp(\text{\robust}) \leq$ \
	\[
	\begin{cases}
		\min\{(1-\eta)/(1-H_n), \sqrt{M/m}\}, & \text{for negative error}\\
		\min\{(1+\eta)/(1-H_n), \sqrt{M/m}\}, & \text{for positive error}.
	\end{cases}
	\]
\end{theorem}

\begin{proof}
	Consider first the case of negative error. 
	Suppose that $1/(1-H_n) \leq \sqrt{M/m}$.
	We have $p' = p(1-H_n) \leq p(1-\eta) = p^*$. 
	Thus, the reservation price $p'$ of \robust is no larger than $p^*$, and the algorithm indeed accepts a price at least as high as $p'$. 
	Therefore, $\comp(\text{\robust}) \leq p^*/p' = \frac{p(1-\eta)}{p(1-H_n)} = (1-\eta)/(1-H_n)$.
	If $1/(1-H_n) > \sqrt{M/m}$, then  $(1+H_p)/(1-H_n) > \sqrt{M/m}$, hence \robust applies \ONstar and  $\comp(\text{\robust})\leq \sqrt{M/m}$. 
	
	Next, we consider the case that the error is positive.
	Suppose that $(1+H_p)/(1-H_n) \leq \sqrt{M/m}$. We have $p' = p(1-H_n) \leq p \leq p(1+\eta)=p^*$. 
	Again, this implies that \robust accepts a price at least as high as $p'$, and thus 
	$\comp(\text{\robust}) \leq p^*/p' = (p(1+\eta))/(p(1-H_n)) = (1+\eta)/(1-H_n)$.
	If $(1+H_p)/(1-H_n) > \sqrt{M/m}$, then \robust applies \ONstar and hence $\comp(\text{\robust}) \leq \sqrt{M/m}$. 
\end{proof}

We also prove a matching bound which establishes that \robust is the optimal non-oblivious algorithm:

\begin{theorem}
	Any non-oblivious algorithm has competitive ratio at least
	\[
	\begin{cases}
		\min\{(1-\eta)/(1-H_n), \sqrt{M/m}\}, & \text{for negative error} \\
		\min\{(1+\eta)/(1-H_n), \sqrt{M/m}\}, & \text{for positive error}.
	\end{cases}
	\]
\end{theorem}

\begin{proof}
	Let $A$ denote a non-oblivious algorithm. Note that $A$ must be robust, otherwise its competitive ratio is larger than $\sqrt{M/m}$ for some value of the error. Let $s'_A(p,m,M,H_n,H_p)\in[m,M]$ be the smallest price for which 
	$A$ accepts on day 1. Define $r_A=s'_A(p,m,M,H_n,H_p)/p$. 
	It must hold that $r_A <1$; otherwise, by Theorem~\ref{th:lowerBound}, $\comp(A)=M/m$ for positive values of error smaller than $r_A-1$, contradicting the robustness of $A$.

Let $\epsilon>0$ be any small value such that $\epsilon  < \sqrt{m/M} \cdot \eta_p$.
	By Theorem~\ref{th:lowerBound}, for a value of negative error equal to $\eta_n = 1-r_A + \epsilon$, the competitive ratio of $A$ must be at least $(1-\eta_n)M/m = (r_A-\epsilon)M/m$. Given that $A$ is robust, it follows that $(r_A-\epsilon)M/m \leq \sqrt{M/m}$, that is, $r_A \leq \sqrt{m/M}+\epsilon$.

	We further use the assumption that $A$ is robust to establish that $r_A \leq 1-H_n$. By way of contradiction, suppose that $H_n> 1-r_A$.
	For a fixed value of positive error $0 < \eta_p \leq H_p$, we have 
\begin{eqnarray*}
		\comp(A) &\geq& \min \{(1+\eta_p)/r_A, M/m\}  \geq \\
		&\min& \{(1+\eta_p)/(\sqrt{m/M}+\epsilon), M/m\} > \sqrt{M/m}.
	\end{eqnarray*}

	Since $r_A \leq 1-H_n$, for any values of negative error $\eta$, we have $\eta \leq 1-r_A$ and by Theorem~\ref{th:lowerBound}, $\comp(A) \geq (1-\eta)/r_A \geq (1-\eta)/(1-H_n)$. For values of positive error $\eta$, $\comp(A) \geq \min\{(1+\eta)/r_A,M/m\} \geq (1+\eta)/(1-H_n)$. 
\end{proof}

\newcommand{\alg}{\textsc{Alg}\xspace}
\newcommand{\Alg}{\alg}
\section{Query-based Predictions}
\label{sec:binary}

In this section, we study the setting in which the prediction is in the form of responses to $n$ binary queries $Q_1, \ldots ,Q_n$, for some fixed $n$. Hence, the prediction $P$ is an $n$-bit string, where the $i$-th bit is the response to $Q_i$. 
We assume that the algorithm knows an upper bound $H$ on $\eta$. 
Therefore, the responses to least $k-H$ queries are guaranteed to be {\em correct}, and the responses to at most $H$ queries may be {\em incorrect} (or {\em wrong}). We assume the setting of non-oblivious algorithms in this model. This is because without an upper bound on the error, the algorithm is in a state of complete lack of knowledge concerning the truthfulness of the queries, and it is not obvious how to use them in a meaningful way.

We present two algorithms, both of which use {\em comparison queries} concerning the best price $p^*$. That is, the queries are in the form of ``Is $p^* \leq b$, for some given value $b$?''. In our first algorithm, the values of $b$ form a strictly increasing sequence, which allows us to narrow $p^*$ within an interval (range) from a candidate set of $n$ intervals. The second algorithm implements a {\em robust} version of binary search, in which the candidate set of intervals is exponential in $n$, hence we can narrow $p^*$ within an interval of much smaller size, and thus obtain a much better estimate on $p^*$.

\subsection{Robust Linear Interval Search algorithm}
Define $m=a_0, a_1, \ldots, a_n=M$ so that $r_n = a_1/a_0 = a_2/a_1 = \ldots = a_n/a_{n-1}$, which implies that $r_n = (M/m)^{1/n}$. 
Define the $n$ intervals $E_1, \ldots, E_n$, where $E_i = [m, a_{i})$ (we have $1\leq i \leq n$). 
Query $Q_i$ asks whether the best price $p^*$ is in $E_i$ or not, and the response denoted by ``1'' or ``0'', respectively. 

Consider the $n$-bit string $P$ formed by responses to $Q_1, \ldots, Q_n$. 
If all responses were correct, then $P$ would consist of $j$ 0s followed by $(n-j)$ 1s for some 
$j\in [1,n]$. This would prove  that $p^*$ is in the range $[a_j,a_{j+1})$. The algorithm then could use $a_j$ as its reservation price, which yields a competitive ratio of at most 
${\frac{ a_{j+1}}{a_j}} = (M/m)^{1/n}$.

We describe the algorithm {\em Robust Linear Interval Search} (\RLIS), that works in the presence of error. 
From the way queries are defined, it may be possible to detect and correct some of the wrong responses in $P$ as follows. 
Suppose the response to $Q_i$ is 1, while the response to at least $H+1$ queries $Q_j$ that come after  $Q_i$ (that is, $j>i$) are 0. 
Given that the number of incorrect responses cannot exceed $H$, we infer that the response to $Q_i$ must be incorrect. With a similar argument, if the response to $Q_i$ is 0 and the responses to at least $H+1$ queries that come before $Q_i$ are 1, then the response to $Q_i$ must be incorrect. 
Thus \RLIS starts with a {\em preprocessing phase} which corrects these two types of incorrect  responses. This results in an updated prediction string $P'$ in which every 1-bit is followed by at most $H$ 0-bits and every 0-bit is preceded by at most $H$ 1-bits.  

If all responses in $P'$ are 0, \RLIS sets its reservation price to $a_{n-H+1}$. Otherwise, let $i_1$ denote the index of the first 1 in $P'$, and let $\alpha \leq H$ be the number of 0s after index $i_1$. Define $l = \max\{0, i_1-(H+1-\alpha)\}$, then \RLIS sets its reservation price to $a_{l}$. 

Algorithm~\ref{alg:RLIS} describes \RLIS in pseudocode. The queries that \RLIS asks and hence their responses can be pre-computed and stored in the prediction array $P$. Note that the algorithm has two phases: a pre-processing phase in which it detects and corrects some of the predictions in $P$ (Lines~\ref{linePreSt} to~\ref{linePreEnd}). The pre-processing phase results in an updated prediction array $P'$ which is used to set the reservation price (Lines~\ref{linePostSt} to~\ref{linePostEnd}). In terms of time complexity, it is straightforward to verify that both phases of the algorithm can be completed in $O(n)$, and therefore \RLIS sets its reservation price in time $O(n)$.

\definecolor{cmntscolor}{rgb}{0.6, 0.6, 0.6}
\algnewcommand{\LineComment}[1]{\ \\ $\triangleright${\color{cmntscolor}{\ #1}}}

\begin{algorithm}[!t]
	\caption{Robust Linear Interval Search (\RLIS)}\label{alg:RLIS}
	\begin{algorithmic}[1]
		\State \textbf{Input:} $m,M$ (lower and upper bounds for the best price); binary string $P = (p_1, p_2, \ldots, p_n)$ of responses; an upper bound $H$ for the number of incorrect answers.  \State \textbf{Output:} a reservation price $rp$   
		
		\LineComment{Preprocessing phase: detect and correct errors}
		\State $P' \gets $ a copy of $P$
		\label{linePreSt}
		\For{$i \gets 1$ to $n$}
		\State $suc(i) \gets $ no. indices $j$ such that $j >i$ and $P[j]=0$
		\State $pre(i) \gets $ no. indices $j$ such that $j <i$ and $P[j]=1$
		\If{ ($i=1$ \& $suc(i) \geq H+1)$ or \\ \hspace*{.8cm}
			($i=0$ \& $pre(i) \geq H+1$)}
		\State $P'[i] \gets 1-P'[i]$ \Comment{\color{cmntscolor}{fix the detected error}\color{black}}
		\EndIf 
		\EndFor 
		\LineComment{Setting the reservation price}
		\label{linePreEnd}\If{($\exists$ an index $i$ such that $P'[i] = 1$)}
		\label{linePostSt}\State $i_1 \gets$ the smallest $i$ so that $P'[i] =1$.
		\State $\alpha \gets suc(i)$ \Comment{{\color{cmntscolor}$(\alpha \leq H)$}}
		\State $l = \max \{0, i_1 - (H + 1 -\alpha)\}$
		\Else
		\State $l = n-H+1$
		\EndIf
		\State  $a_{l-1} \gets m \cdot (M/m)^{(l-1)/n}$	\Comment{{\color{cmntscolor}{the reservation price}}} \\		\Return $a_{l-1}$ \label{linePostEnd}
	\end{algorithmic}
\end{algorithm}


\begin{theorem}\label{th:main:linear}
	Algorithm \RLIS has competitive ratio at most $(M/m)^{2H/n}$.
\end{theorem}
\begin{proof}
	First suppose all responses in $P'$ are 0. There can only be a suffix of at most $H$
	0-responses in $P'$ which are incorrect, that is, $p^*$ is 
	in the range $[a_{n-(H-1)},M]$.
	 Given that \RLIS has reservation price $a_{n-(H-1)}$, and $p^* \leq M$, 
	 its competitive ratio is at most $M/a_{n-(H-1)} =  (M/m)^{(H-1)/n}$.

	Next, suppose that $P'$ contains a 1-bit. We consider two cases, depending on the presence or absence of a 0-bit in $P'$ after index $i_1$. If there is no 0-response after index $i_1$ (that is $\alpha=0$), then $p^* \leq a_{i_1+H-1}$ (because at most $H$ queries can be wrong), while the profit of \RLIS is at least $a_{l}$. The competitive ratio is thus at most $a_{i_1+H-1}/a_l \leq (M/m)^{2H/n}.$
	
	Next, suppose that there is a 0-response after index $i_1$, and let $j_0>i_1$ denote the index of the last such 0 in $P'$. We will show that $p^* \geq a_l$, where $l$ is the reservation price of
	\RLIS. Suppose first that the response to $Q_{i_1}$ is incorrect. Then, $p^* \geq
	a_{i_1}$. Suppose next that the response to $Q_{i_1}$ is correct. In this case, the  $\alpha$ 0-responses after index $i_1$ must be wrong. Since there can be up to $H$ errors, at most $(H-\alpha)$ 0-responses that immediately precede index $i_1$ can be wrong. Therefore, $p^* \geq a_{l}$ where $l = \max\{0, i_1-(H+1-\alpha)$\}. This also implies that \RLIS has profit at least $a_l$.

To finish the proof, we need an upper bound on $p^*$. If the response to $Q_{j_0}$ is incorrect, then $p^* \leq a_{j_0}$. Otherwise, there are $j_0-i_1-\alpha$ wrong 1-responses before $j_0$, from the definition of $j_0$. Therefore, up to $H -(j_0 -i_1- \alpha)$ 1-responses that follow $j_0$ can also be wrong. That is, $p^*$ can be as large as $a_{j'}$, where $j' = \min \{n, j_0 + H - (j_0 -i_1- \alpha)\} = \min \{n, H +i_1+\alpha\} $. 
	Given that the reservation price is $a_l$, the competitive ratio of the algorithm is therefore at most ${a_{j'}/a_{l}} \leq (r_n)^{2H} = (M/m)^{2H/n}$.
	%
	%
	%
\end{proof}

\subsection{Robust Binary Interval Search algorithm}	 
\label{subsec:robust.binary}

Algorithm \RLIS uses the queries so as to select a reservation price from $n$ candidate intervals. We will now show how to increase this number to $2^n$ using a new {\em Robust Binary Search} algorithm 
(\RBIS). Partition the interval $[m,M]$ into $2^n$ intervals $L_1, L_2, \ldots, L_{2^n}$, where $L_i = (a_{i-1},a_{i}]$, $a_0=m$, and $a_{2^n} = M$. 
We define the $a_i$'s so that $\rho = a_1/a_0 = a_2/a_1 = \ldots = a_{2^n}/a_{2^{n-1}}$, where $\rho =(M/m)^{1/2^n}$. 

Suppose that $L_1, \ldots ,L_{2^n}$ correspond to the $2^n$ leaves of a binary tree $T$ of height $n$, and that the best price $p^*$ is in the interval $L_x$ for some $x\in[1,2^n]$. With perfect queries (zero error), it is possible to find $L_x$ using binary search on $T$, which leads
to a competitive ratio  $a_x/a_{x-1} =  (M/m)^{1/2^n}$, by choosing a reservation price equal to 
$a_{x-1}$. This is the approach of~\cite{clemente2016advice}.
Unfortunately this simple approach is very inefficient even if a single error occurs (e.g., if $Q_1$ receives a wrong response, the search will end in a leaf $L_y$, where $|x-y|$ is as large as $2^{n/2}$.)

Searching with erroneous queries is a well-studied topic, see e.g., the book~\cite{DBLP:series/eatcs/Cicalese13}. A related problem to ours was studied in~\cite{RivestMKWS80} and~\cite{DisserK17}, however there are important differences with respect to our setting. First, these works consider a ``dual'' problem to ours in which the objective is to minimize the number of queries so as to locate an {\em exact} leaf in the binary tree. Second, there are certain significant implementation issues that need to be considered. Specifically,~\cite{DisserK17} assumes that  
 when reaching a leaf, an oracle can respond whether this is the sought element or not (in other words, the algorithm receives an error-free response to a query of the form ``is an element exactly equal to $x$''). For our problem and, arguably, for many other problems with query-based predictions, this assumption cannot be made. Moreover~\cite{RivestMKWS80} does not have an efficient implementation, specifically in comparison to~\cite{DisserK17}. We propose a new algorithm using some ideas of~\cite{DisserK17} that is applicable to our problem, and has an efficient implementation.



\paragraph{Algorithm description}
Recall that $T$ is a binary search tree with leaves $L_1, L_2, \ldots, L_{2^n}$ and 
that we search for the leaf $L_x$. We denote by 
$l(v)$, $r(v)$ the left and right child of $v$, respectively, and by $T_v$ the subtree rooted at $v$.

We describe the actions of the algorithm. Suppose the algorithm is at node $v$ at the beginning of iteration $i$ (in the first iteration, $v$ is the root of $T$). The algorithm first asks a {\em main query}, defined as follows:  ``Is $x \leq q$?", where $q$ is such that $L_q$ is the rightmost leaf of the left subtree of $v$. We denote by  {\tt main(v)} the response to this query.  As we discuss shortly, the search may visit the same node multiple times, so we emphasize that {\tt main(v)} is the response to the most recent main query at $v$.
Next, the algorithm finds the first ancestor of $v$ in $T$, say $w$, for which {\tt main(w)}$\neq${\tt main(v)}. We denote this ancestor of $v$ by $anc(v)$, if it exists, and define $anc(v)=\emptyset$, otherwise. The algorithm continues by asking a {\em checkup query} which is a repetition of the main query asked for $w$. We denote the response to the checkup query as {\tt check(v)}. The algorithm continues by taking one of the following actions, after which iteration $i+1$ begins: 
\begin{itemize}
	\item \textbf{Move-down:} If  $anc(v)=\emptyset$ or {\tt check(v)} = {\tt main($anc(v)$)}, \RBIS moves one level down in $T$. That is, if {\tt main(v)} is Yes (respectively No), \RBIS moves to $l(v)$ (respectively $r(v)$).
	\item \textbf{Move-up}: If {\tt check(v)}$\neq${\tt main($anc(v)$)}, \RBIS moves one level up to the parent of $v$. In this case, \RBIS increments a counter {\tt mu}, which is originally set to $0$. 
\end{itemize}

\newcommand{\mend}{\ensuremath{mu_{end}}\xspace}
The algorithm continues as described above until it exhausts its number $n$ of queries. Suppose the search stops at some node $u$, and let $a_u$ denote the $(H-{ mu_{end}})$-th ancestor of $u$ (or the root, if such an ancestor does not exist), where \mend is the content of {\tt mu} at the end of the search.
Let $L_l$ be the leftmost leaf in $T_{a_u}$, i.e the subtree rooted at $a_u$. 
Then \RBIS returns this leftmost leaf in $T_{a_u}$. In particular, for the online search problem, the algorithm sets its reservation price to $a_{l-1}$. 

Algorithm~\ref{alg:RBIS} describes \RBIS in psuedocode. Queries of \RBIS depend on the location of the search node in the search tree $T$, which indeed depends on the errors in previously responded queries. As such, unlike \RLIS, it is not possible to provide the responses to all queries in advance. Therefore, we assume \RBIS has access to a \emph{response oracle} that answers its queries in real time.  The algorithm has an initializing phase (Lines~\ref{lineInitSt} to~\ref{LineInitEnd}), a search phase where it applies Moves-down and Moves-up operations in the search tree (Lines~\ref{lineSearchSt} to~\ref{lineSearchEnd}) and a final phase where it computes the reservation price (Lines~\ref{LineFinalSt} to~\ref{LineFinalEnd}). All these phases take $O(n)$ time, and therefore \RBIS sets its reservation price in time $O(n)$.

{
\begin{algorithm}[H]
	\caption{Robust Binary Interval Search (\RBIS)}\label{alg:RBIS}
	\begin{algorithmic}[1]
		\State \textbf{Input:} $m,M$ (lower and upper bounds for the best price); A response oracle $OR$; an upper bound $H$ for the number of incorrect answers.  
		\State \textbf{Output:} a reservation price $rp$ 
		\LineComment{Initializing}
		\State $T\gets$ a full binary~tree~with~$2^n$~leaves~$L=~(L_1, \ldots, L_{2^n})$
		\For{every node $x \in T$}\label{lineInitSt}
		\State $\mathtt{main(x)} \gets -1$ 
		\EndFor 
		\State $uq \gets 0$ \Comment{{\color{cmntscolor}{no. used queries}}}
		\State $mu \gets 0$ \Comment{{\color{cmntscolor}{no. Move-up operations}}}
		\State $v \gets$ root of $T$ \label{LineInitEnd}
		\LineComment{Searching in the tree} \label{lineSearchSt}
		\While{$uq \leq k$}
		\State $T_{\text{left}_v} \gets $ subtree of $T$ rooted at the left child of $v$
		\State $q_v \gets $ index of the right-most leaf of $T_{\text{left}_v}$ in $L$
		\State $Q_{uq} \gets $ ``is $p^* \leq m \cdot (M/m)^{q_v/2^n}$?" \Comment{{\color{cmntscolor}{main query of $v$}}}
		\State $\mathtt{main(v)} \gets OR.$response$(Q_{uq})$
		\State $uq \gets uq +1$
		\State $w \gets $ the first ancestor of $v$ s.t. $\mathtt{main(w)} \neq \mathtt{main(v)}$  \Comment{{\color{cmntscolor}{$w$ is $anc(v)$}}}		
		\If{$w\neq \phi$}
		
		\State $T_{\text{left}_w} \gets $ subtree of $T$ rooted at the left~child~of~$w$
		\State $q_w \gets $ index of the right-most leaf~of~$T_{\text{left}_w}$~in~$L$
		\State $Q_{uq} \gets $ ``is $p^* \leq m \cdot (M/m)^{q_w/2^n}$?" \Comment{{\color{cmntscolor}{the checkup query at $v$}}}
		\State $\mathtt{check(v)} \gets OR.$response$(Q_{uq})$
		\State $uq \gets uq +1$
		\EndIf
		\If{($w = \phi$) or ($\mathtt{check(v)}$ = $\mathtt{main(w)}$ )}
		\If {$\mathtt{main(v) = }$ ``Yes"} \Comment{{\color{cmntscolor}{Move-down operation}}}
		\State $v \gets $ left child of $v$
		\Else
		\State $v \gets $ right child of $v$
		\EndIf
		\Else 
		\State $v \gets $ parent of $v$ \Comment{{\color{cmntscolor}{Move-up operation}}}
		\State $mu \gets mu + 1$
		\EndIf	
		\EndWhile \label{lineSearchEnd}
		\LineComment{Search ends; setting the reservation price}\label{LineFinalSt}
		\State $a_u \gets$ the $(H-mu)$'th ancestor of $u$
		\State $T_{a_u} \gets $ the tree rooted at $a_u$
		\State $L_l \gets$ the index of the leftmost leaf of $T_{a_u}$ in $L$
		\State $a_l \gets m \cdot (M/m)^{l/2^n}$\Comment{{\color{cmntscolor}{set the reservation price}}} \\
		\Return $a_l$\label{LineFinalEnd}
	\end{algorithmic}
	
\end{algorithm}
}


\paragraph{Analysis} We first show the following useful lemmas. 

\begin{lemma}
	Suppose a Move-down operation takes place at node $v$, and let $ch(v)$ denote the child of $v$ to which the search moves. Then either $L_x$ is in $T_{ch(v)}$ or at least one of the responses {\tt main(v)} and {\tt check(v)} are incorrect.
	\label{lemma:helper.claim} 
\end{lemma}
\begin{proof}
	If $anc(v)=\emptyset$,  
	then all previous main queries have received the same response (involving Move-down operations, either all to the left or all to the right). Therefore, if {\tt main(v)} is correct, then for every node $y$ on the path from the root to $v$, {\tt main(y)} is also correct and 
	$L_x$ is in $T_{ch(v)}$, hence the lemma follows.

	Next, suppose that {\tt main(v)} is correct, but $L_x \neq T_{ch(v)}$. Since the algorithm moves down, it must be that {\tt main(anc(v))}={\tt check(v)}. To prove the lemma, it suffices to show these two responses are wrong. Without loss of generality, suppose that  {\tt main(v)} is Yes and  {\tt main(anc(v))} is No 
	(the opposite case is handled symmetrically). 
	Given that {\tt main(v)} is correct, $L_x$ must be either in $T_{ch(v)}$ or in the left subtree of 
	$anc(v)$. In the former case, the lemma follows directly. In the latter case, {\tt main($anc(v)$)}, which is precisely {\tt check(v)},
	is incorrect, and thus the lemma again follows.
\end{proof}

The proof of the following lemma is based on Lemma~\ref{lemma:helper.claim}, by showing that the search ends sufficiently deep in the tree.

\begin{lemma}\label{lemma:twoclaims} 
	The following hold: 
(i) Node $a_u$ is at depth at least $\lfloor n/2\rfloor-2H$ in $T$; and
(ii)  $L_x$ is a leaf of $T_{a_u}$.
\end{lemma}

\begin{proof}
	To prove (i), note that since there are $n$ queries, and each iteration invokes up to two queries, the number of iterations is at least $\lfloor n/2 \rfloor$. Among these iterations, \mend 
	of them are Move-up iterations and the remaining $\lfloor n/2 \rfloor-\mend$ are Move-down iterations. 
	Therefore, the search ends at a node $u$ of depth $\lfloor n/2 \rfloor - 2\mend$. Given that $a_u$ is the $H-\mend$'th ancestor of $u$, its depth is at least $\lfloor n/2 \rfloor - 2\mend - (H-\mend) \leq \lfloor n/2 \rfloor - 2H $. The last inequality holds because $\mend\leq H$. 
	
	We prove (ii) by way of contradiction. Suppose that $L_x$ is not in $T_{a_u}$. 
	Recall that $a_u$ is the $(H-\mend)$'th ancestor of node $u$. Therefore, the algorithm must have made at least $H-\mend+1$ Move-down operations, in a subtree that does not contain $L_x$. 
	From Lemma~\ref{lemma:helper.claim}, any of these operations include at least one incorrect response to their main or checkup query, resulting in at least $H-\mend+1$ incorrect responses for iterations with Move-down operations on the search path from $a_u$ to $u$. In addition, each Move-up query is associated with a wrong response. To see that, suppose there is a Move-up query at node $v$. If $L_x$ is not in $T_v$, then the main query at the parent of $v$ has been incorrectly answered. Otherwise, if $L_x$ is in $T_v$, then $\texttt{check(v)}$ is incorrectly responded (Move-up operation implies $\texttt{check(v)} \neq \texttt{main(anc(v))}$ and $\texttt{main(anc(v))}$ is correctly answered because $L_x \in T_v$). We conclude that, in addition to the $H-\mend+1$ incorrect responses for iterations with Move-down operations, there are $\mend$ incorrect responses associated with the Move-up queries. Therefore, the total number of wrong responses to queries must be at least $H+1$, contradicting the fact that the number of wrong responses is at most $H$. 
\end{proof}

\begin{theorem}\label{th:main:binary} 
	For every $H\leq n/4$, \RBIS has competitive ratio at most  $(M/m)^{2^{2H-n/2}}$.
\end{theorem}

\begin{proof}
	Let $L_l = [a_{l-1},a_l)$ and $L_r = [a_{r-1},a_r)$ denote the leftmost and rightmost leaves in the subtree rooted at $a_u$. Recall that the algorithm selects $a_{l-1}$ as its reservation price, while Lemma~\ref{lemma:twoclaims} guarantee ensures that $L_x$, and thus $p^*$ is located in the subtree rooted at $a_u$, that is, $p^* < a_r$. Therefore, the competitive ratio of \RBIS is at 
	most $a_r/a_{l-1} = \rho^{r-l+1}$. Moreover, by Lemma~\ref{lemma:twoclaims}, since $a_u$ is at depth at least $d = \lfloor n/2 \rfloor - 2H$ of $T$, the number of leaves in the subtree rooted at $a_u$ is at least $2^{n-d} < 2^{n/2+2H}$, and thus $\comp(\RBIS) \leq \rho^{2^{n-d}} 
	< (M/m)^{2^{n/2 + 2H}/2^n} = (M/m)^{2^{2H-n/2}}$.
\end{proof}

\paragraph{Lower bounds}
We can complement Theorem~\ref{th:main:binary} with the following impossibility result, assuming  comparison-based queries over a binary search tree.

\newcommand{\dual}{\textsc{Dual}\textsc{Search}($m$)\xspace}
\begin{theorem}\label{th:lowerBinary}
	The competitive ratio of any online search algorithm with $n$ comparison-based queries over a binary tree
	is at least $(M/m)^{2^{2H-n}}$, assuming $n\geq 11$.
\end{theorem}

\begin{proof}
	For the sake of contradiction, suppose there is an algorithm $A$ that achieves a competitive ratio better than $\rho = (M/m)^{2^{2H-n}}$. Consider the following search problem that we call \dual: an adversary selects an integer $x$ so that $1 \leq x \leq 2^m$ , and the goal is to find $x$ using a minimum number of queries, out of which up to $H$ queries are incorrectly answered. ~\cite{DisserK17} proved that one cannot solve \dual using less than $m + 2H$ queries 
	(and this holds even if algorithms can receive error-free responses to ``=" queries). 
	
	Let $Q$ be an instance of \dual that asks for $x= x_0$ in a search space of size $m=n-2H$. We show that $A$ can be used to solve \dual. For that, we form an instance $Q'$ of the search problem in which the best value $p^*$ is defined as follows. Partition the interval $[m,M]$ into $2^{n-2H}$ intervals $L_1, L_2, \ldots, L_{2^{n-2H}}$, where $L_i = (a_{i-1},a_{i}]$, $a_0=m$, and $a_{2^n} = M$. We define the $a_i$'s so that $\rho = a_1/a_0 = a_2/a_1 = \ldots = a_{2^{n-2H}}/a_{2^{n-2H-1}}$. Now, let $p^* = a_{x_0}$. In order to solve the instance $Q$ of \dual, we apply $A$ on the instance $Q'$ of the search problem, assuming it is allowed to ask $n$ queries. Let $r_A$ denote the smallest price revealed on day 1 such that $A$ accepts that price on day 1. Suppose $r_A$ is in interval $L_y$ ($y\in [1..2^{n-2H}]$). Then, return $y$ as the answer to the instance $Q$ of \dual. Given the upper bound for $\comp(A) \leq \rho$, it must hold that $x_0 = y$. Otherwise, if $y<x_0$, 
	we will have $\comp(A) \geq a_{x_0}/r_A \geq  a_{x_0}/a_{x_0-1} > \rho$, and if $y> x_0$, 
	we will have $\comp(A) \geq a_{x_0}/m \geq \rho^2$. 
	
	To summarize, we showed that if an online search algorithm $A$ asks $n$ questions (out of which up to $H$ are answered incorrectly) and has a competitive ratio better than $\rho$, then $A$ can be used to solve an instance of $\dual$ on a search space  $m=n-2H$. This, however, contradicts the result of~\cite{DisserK17} and we can conclude that $A$ cannot achieve a competitive ratio better than $\rho$. 
\end{proof}

\newcommand{\bitcoin}{\texttt{Bitcoin-to-USD}\xspace}
\newcommand{\ether}{\texttt{Ethereum-to-USD}\xspace}
\newcommand{\jpy}{\texttt{Yen-to-CAD}\xspace}
\newcommand{\euro}{\texttt{Euro-to-USD}\xspace}


\section{Experimental evaluation}
\label{sex:experimental}

\subsection{Benchmarks and input generation} \label{subsec:benchmarks}
We evaluate our algorithms on benchmarks generated from real-world currency exchange rates, which are publicly available on several platforms. Specifically, we rely on~\cite{tradingaacademy}. We used two currency exchange rates (\texttt{Bitcoin-to-USD} and \texttt{Ethereum-to-USD}) and two fiat currency exchange rates (\texttt{Euro-to-USD} and \texttt{Yen-to-CAD}). 
In all cases, we collected the closing daily exchange rates for a time horizon starting on January 1st, 2018 and ending on September 1st, 2021, which we use as the daily prices. 

For each benchmark, 20 instances $I_1, I_2, \ldots, I_{20}$ of the online search problem are generated as follows. We select 20 {\em starting days} from the time horizon so that consecutive starting days are evenly distanced. 
Each starting day and the 199 days that follow it form an instance (of length 200) of the search problem. For each such instance, we select $m$ and $M$ to be respectively the minimum and maximum exchange rates. In all experiments, the reported profits are the average taken over these 20 instances. In particular, we use the average profit of the optimal online algorithm \ONstar (without any prediction) as the baseline for our comparisons. 
Similarly, the average value of the best prices (over all instances) is reported as an upper bound for attainable profits.

\subsection{Algorithms with Best-Price Prediction}
We test our algorithms using several values of prediction error.  
For the \bitcoin and \ether benchmarks, we select 500 values of negative error equally distanced in $[0,0.5]$, as well as 500 equally-distanced values of positive error in $[0,0.5]$. 
For the \euro and \jpy benchmarks, we select 500 values of negative/positive error in a smaller range in $[0,0.04]$. This smaller range is consistent with the fact that fiat currencies are substantially less volatile than cryptocurrencies. That is, the values of $m$ and $M$ are very close in instances generated from the fiat currencies. This implies that the range of error ($[0,(M-m)/M]$ and $[0,(M-m)/m]$ for negative and positive errors, respectively) is much smaller for fiat currencies. 

For each selected value, say $\eta_0$, and for each instance $I_x$ of the problem, we test our algorithms for prediction error equal to $\eta_0$, that is, the predicted value is generated by applying the error $\eta_0$ on the best price in $I_x$. The average profit of the algorithm over all instances is reported as its average profit for $\eta_0$.   
Choosing $\eta\leq 0.5$ implies that the prediction $p$ is at least half and at most twice the best price. For real data, such as currency exchange prices, this range of error is sufficient to capture all instances. 


\paragraph{Oblivious algorithms}
For the \bitcoin and \ether benchmarks, we evaluate \ORA with different values of the parameter $r \in \{0.5, 0.75, 1.0, 1.25, 1.5 \}$ and for the \euro and \jpy benchmarks, we set $r\in \{0.96, 0.98, 1.00,1.02, 1.05\}$. Given that the range of error is smaller in the fiat currencies, the reservation price must be closer to the predicted value for $p^*$, that is, $r$ should be closer to 1. 
(recall that $rp$ is the reservation price of an algorithm in this class). Figure~\ref{fig:bitcoinoblivious} illustrates the average profit for instances generated from different benchmarks. 
The findings are consistent with Theorem~\ref{lem:up}. Specifically,
for positive error, for all reported values of $r$, \ORA degrades with $\eta$ (consistently with the linear increase in the competitive ratio in Theorem~\ref{lem:up}). For small values of negative error, the average profit increases by $\eta$, followed by a ``drop" when $\eta$ takes a certain larger value (e.g., when $\eta$ becomes $0.251$ for the algorithm with $r=0.75$). This follows precisely Theorem~\ref{lem:up}, as illustrated in Figure~\ref{fig:crChart}. For larger values of negative error, the algorithms gain a fixed profit (e.g., 15890 in the \bitcoin benchmark), which is the average value of the last-day price. For these values of error, the algorithm sets a reservation price that is too large, and results in the player accepting the last-day price. Last, we note that, as predicted by our competitive analysis, no algorithm dominates another in terms of collected profit. 

\begin{figure}[!t]
	\centering
	\begin{subfigure}[b]{.49\columnwidth}
		\includegraphics[width = \columnwidth]{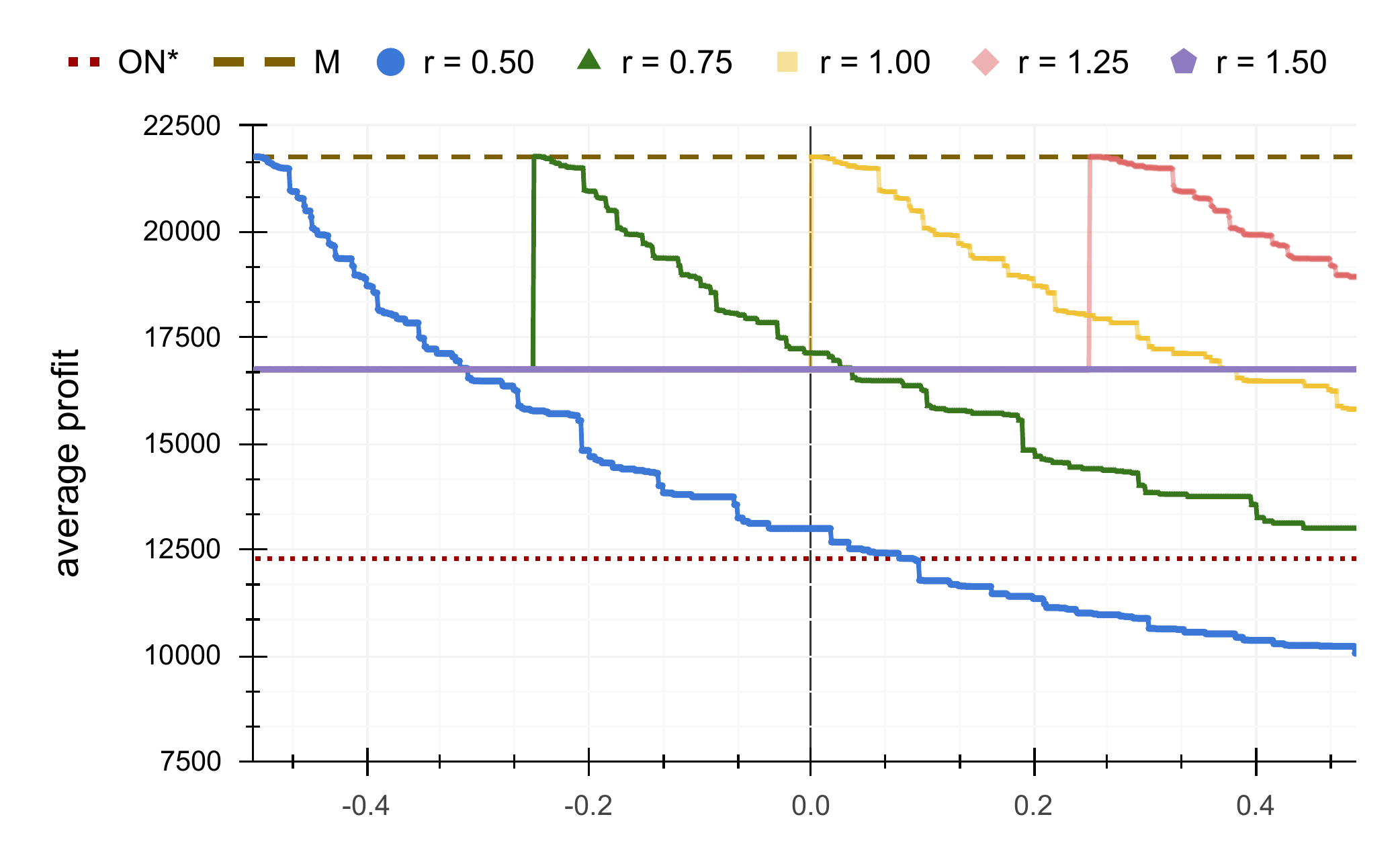}
		\caption{average profit for the \bitcoin benchmark}
	\end{subfigure} \hfill
	\begin{subfigure}[b]{.49\columnwidth}
		\includegraphics[width = \columnwidth]{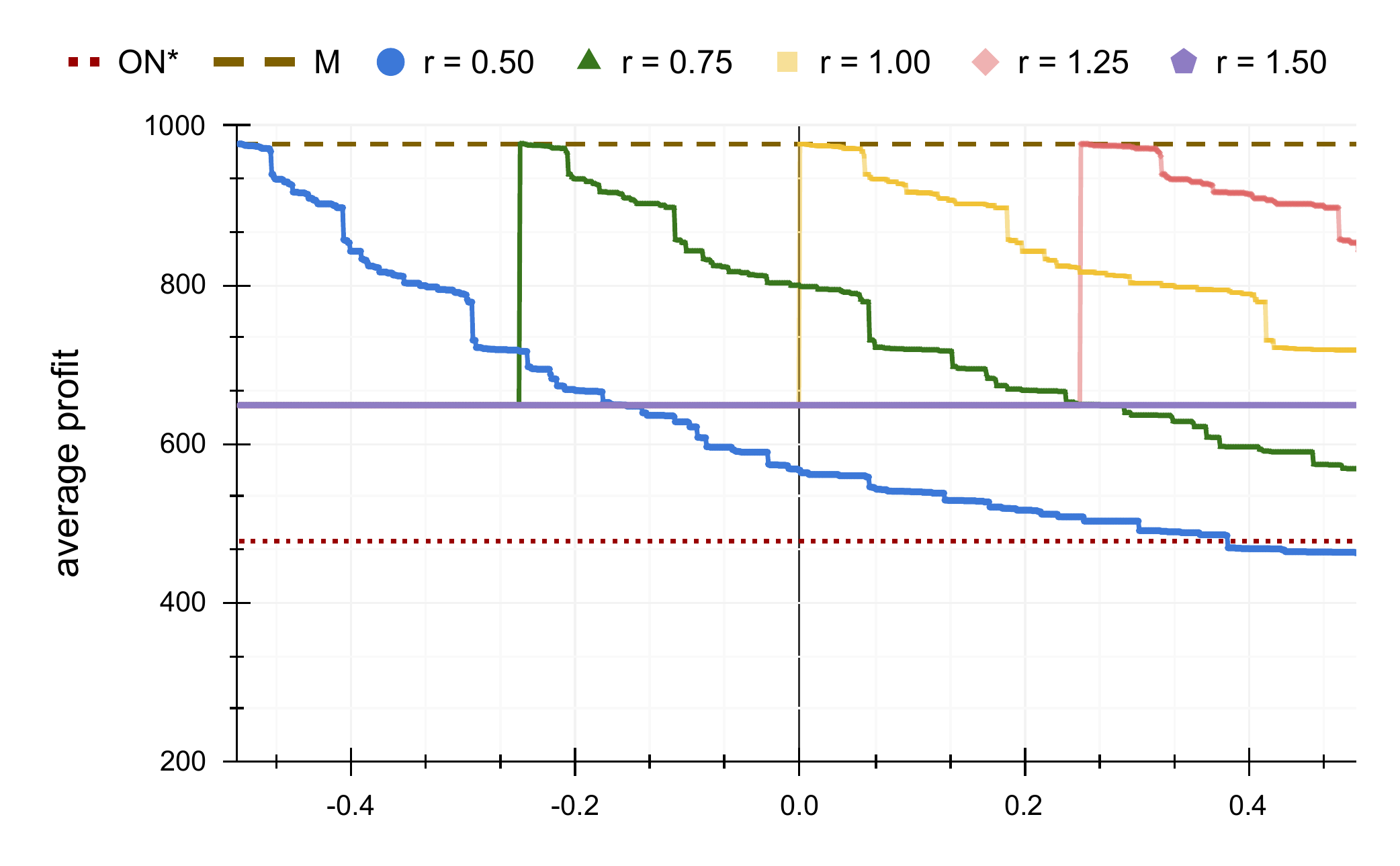}
		\caption{average profit for the \ether benchmark}
	\end{subfigure}
	\hfill
	\begin{subfigure}[b]{.49\columnwidth}
		\centering
		\includegraphics[width = \columnwidth]{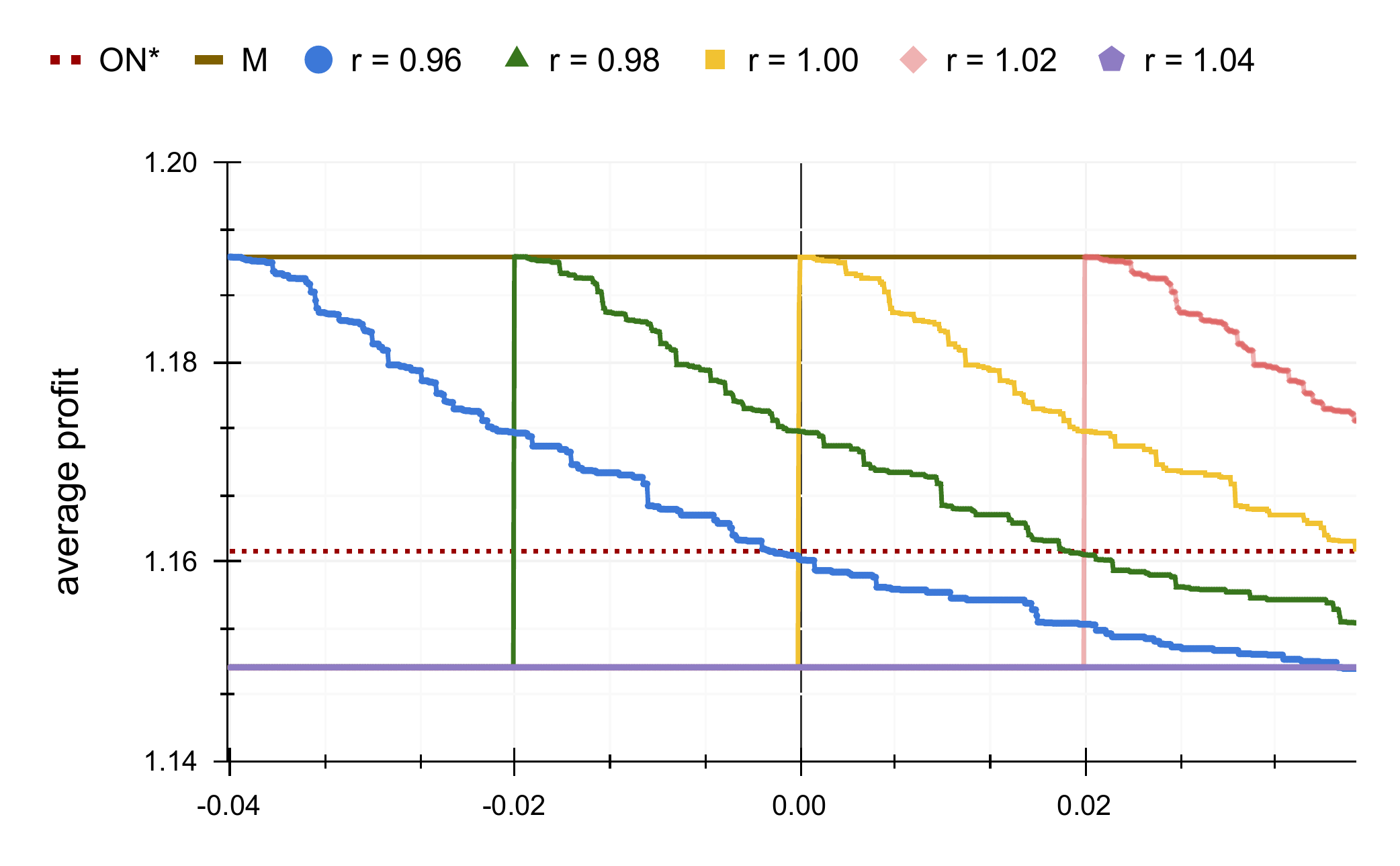}
		\caption{average profit for the \euro benchmark}
	\end{subfigure}
	\hfill
	\begin{subfigure}[b]{.49\columnwidth}
		\centering
		\includegraphics[width = \columnwidth]{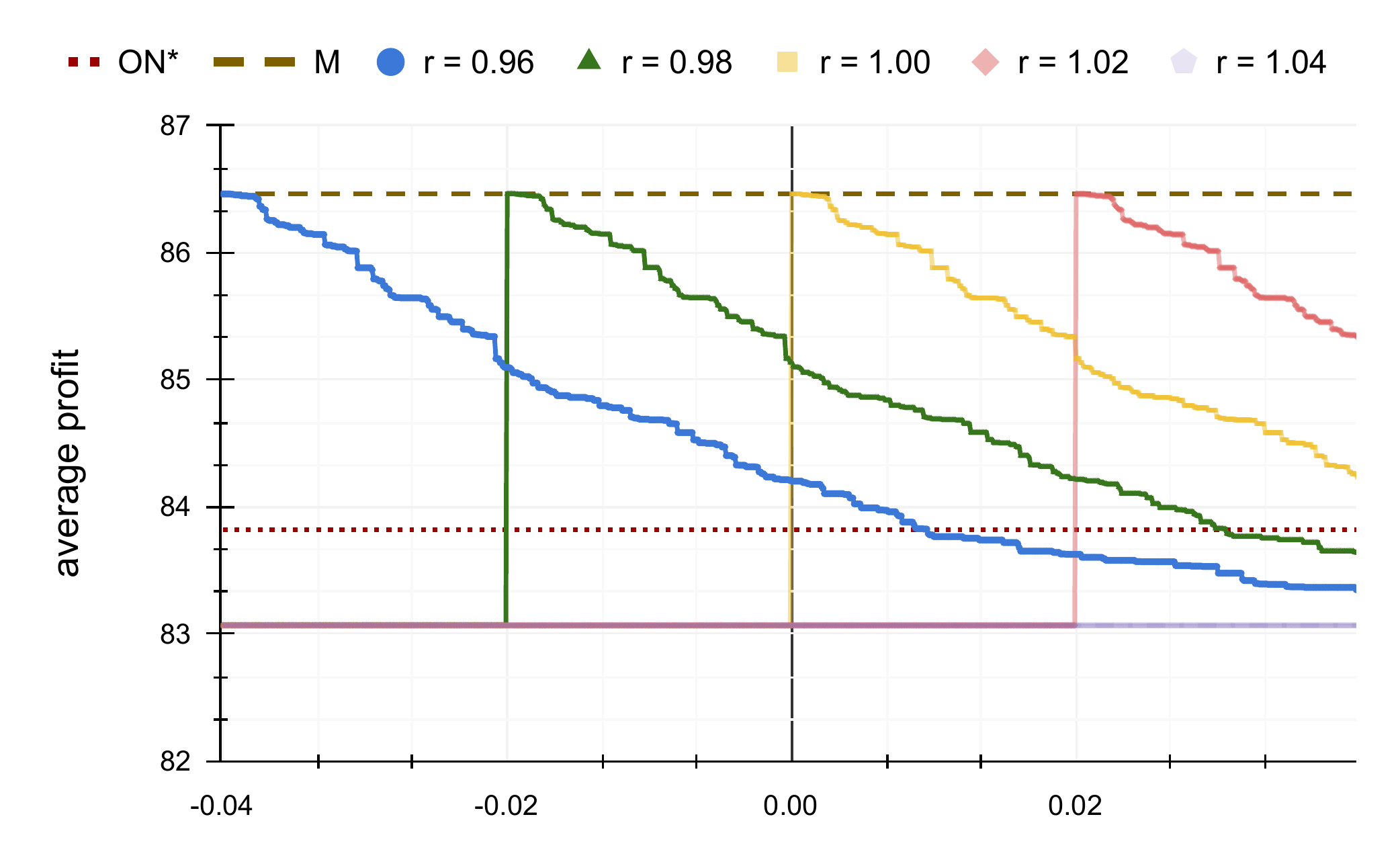}
		\caption{average profit for the \jpy benchmark}
	\end{subfigure}
	\caption{The average profit of \ORA with different values of $r$ over instances generated from different benchmarks.}
	\label{fig:bitcoinoblivious}
\end{figure}

The results demonstrate that predictions about best price lead to profit gains even for oblivious algorithms. 
In particular, all algorithms result in better profit when compared to \ONstar, as long as $\eta <0.5$ (for the \bitcoin benchmark), $\eta < 0.2$ (for the \ether benchmark) and $\eta <0.02$ (for the \euro and \jpy benchmarks).

\paragraph{Non-oblivious algorithms}
We tested \robust with upper bound $H$ on both the positive and negative error. For the \bitcoin and \ether benchmarks, we set $H = H_n = H_p$ for $H \in \{0.1, 0.2, 0.3,  0.4, 0.5 \}$. For the \euro and \jpy benchmarks, we set  $H = H_n = H_p$ for $H \in \{0.005, 0.01, 0.02,  0.03, 0.04 \}$. The smaller range of $\eta$ in fiat currencies implies that we need to test smaller values of $H$.

For each such value of $H$, and for each selected error $\eta$,  we report the average profit over the 20 instances from 
different benchmarks. 
Since the setting is non-oblivious, we only report profits for $\eta \leq H$.
Figure~\ref{fig:ApRobuts} illustrates the average profit for instances generated from different benchmarks. The results are consistent across all benchmarks. We observe that all algorithms improve as the negative error increases and they degrade as the positive error increases. This is consistent with Theorem~\ref{th:haware}. Algorithms with smaller $H$ have an advantage over those with larger $H$, again consistently with Theorem~\ref{th:haware}. These results demonstrate that non-oblivious algorithms can benefit from best-price predictions in all benchmarks. 


\begin{figure}[!t]
	\centering
	\begin{subfigure}[b]{.49\columnwidth}
	\includegraphics[width = \columnwidth]{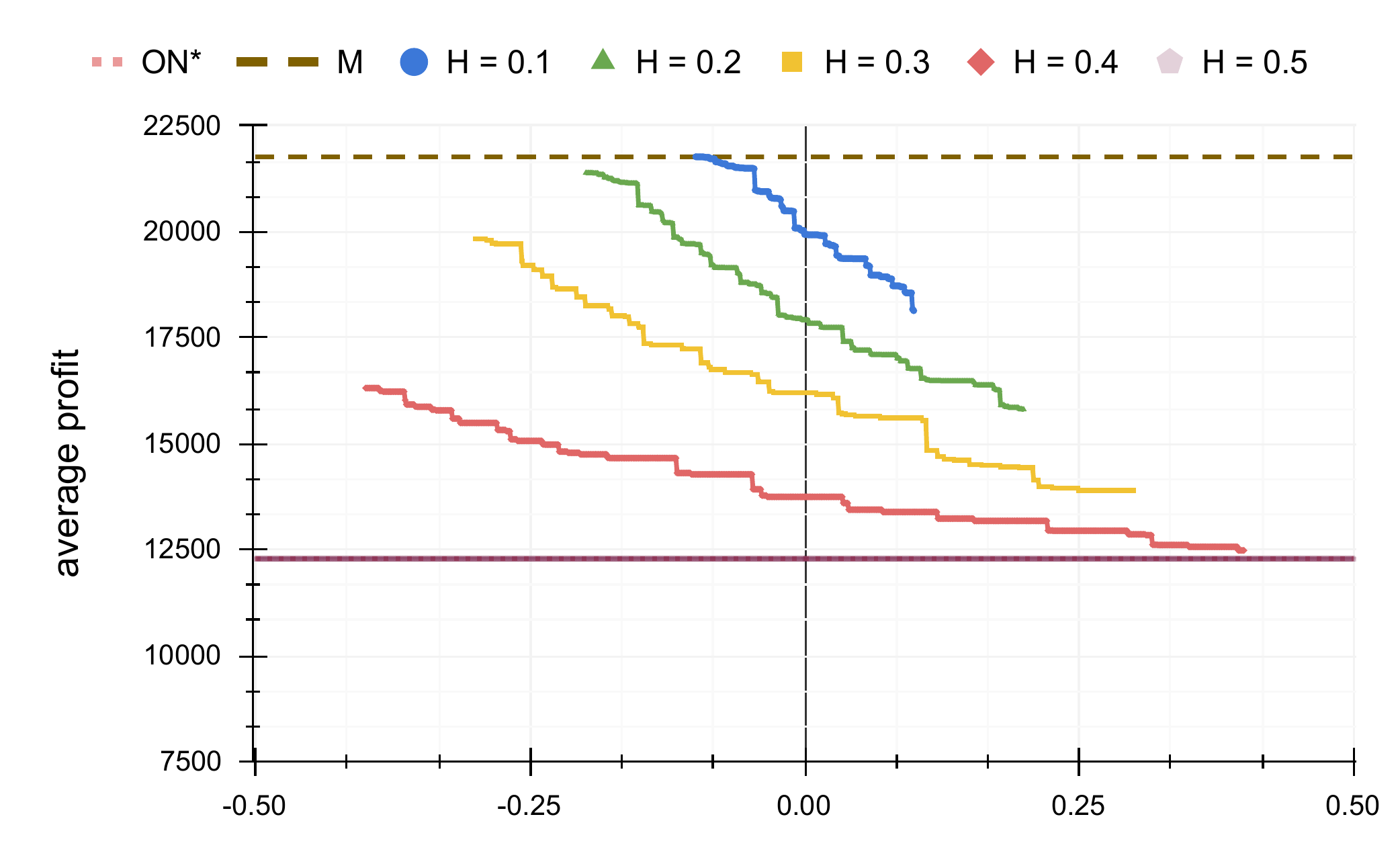}
	\caption{average profit for the \bitcoin benchmark}
	\end{subfigure}
\hfill
	\begin{subfigure}[b]{.49\columnwidth}
		\includegraphics[width = \columnwidth]{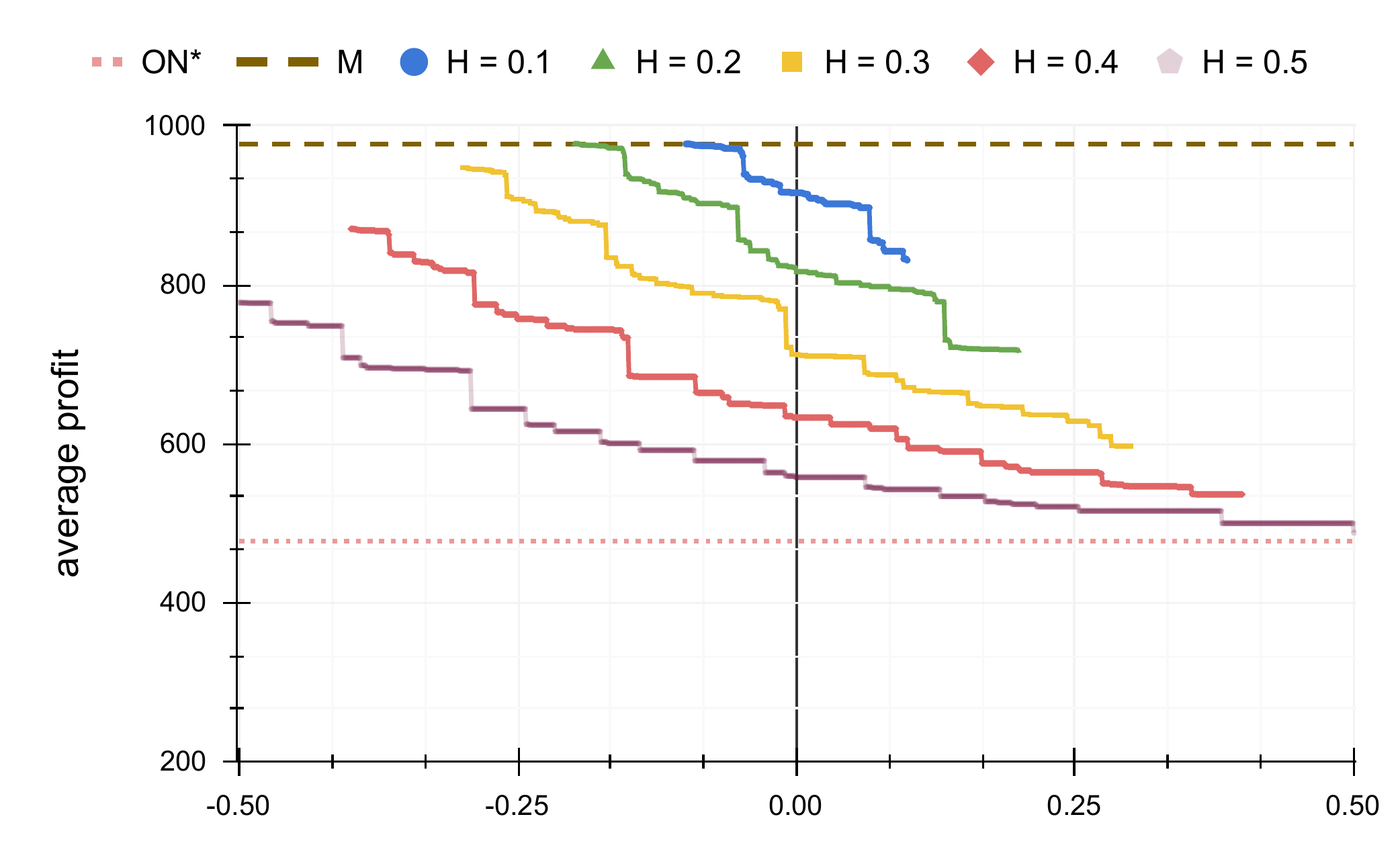}
		\caption{average profit for the \ether benchmark}
	\end{subfigure}
	\hfill
	\begin{subfigure}[b]{.49\columnwidth}
		\centering
		\includegraphics[width = \columnwidth]{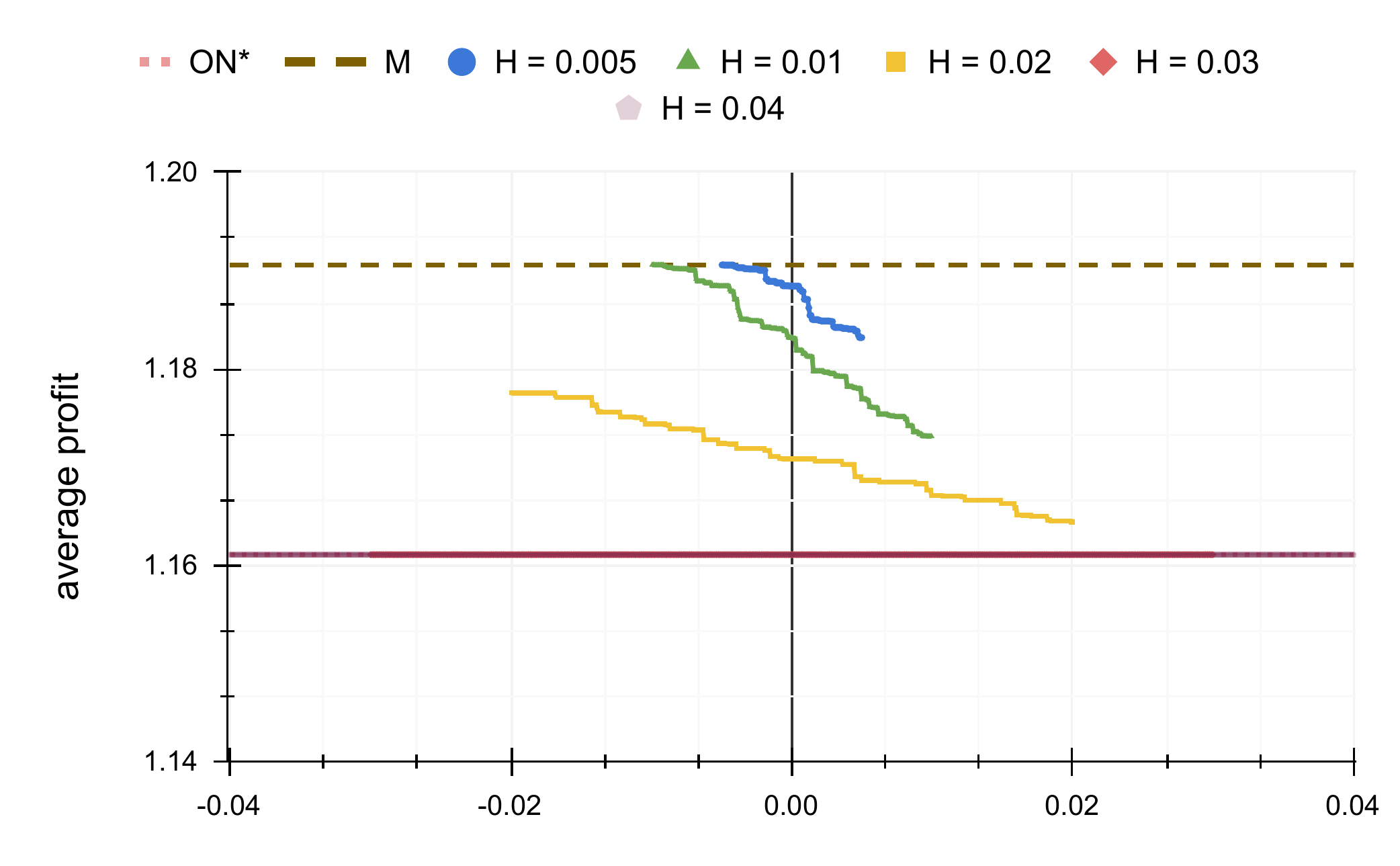}
		\caption{average profit for the \euro benchmark}
	\end{subfigure}
	\hfill
	\begin{subfigure}[b]{.49\columnwidth}
		\centering
		\includegraphics[width = \columnwidth]{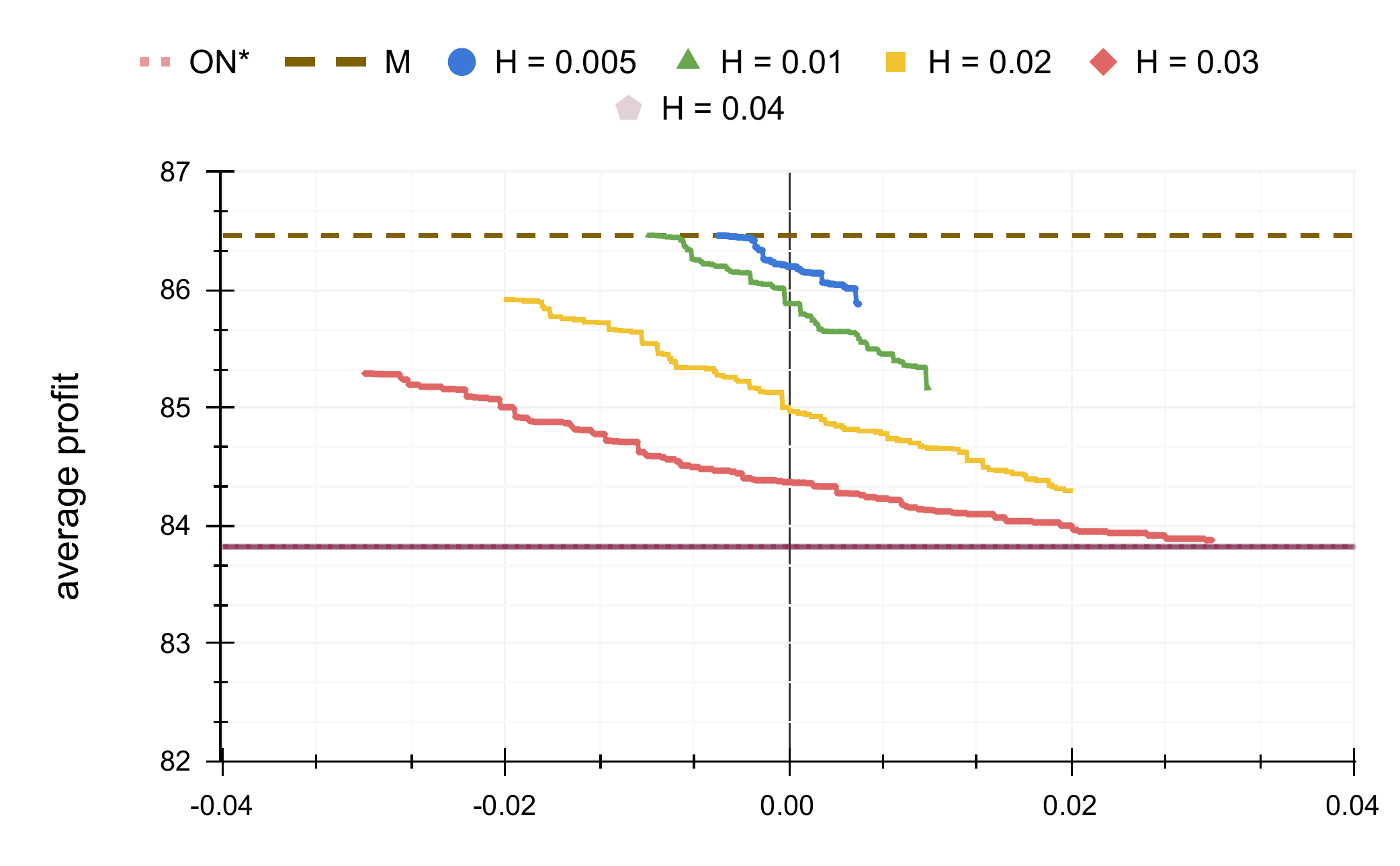}
		\caption{average profit for the \jpy benchmark}
	\end{subfigure}
	\caption{The average profit of \robust with different values of $H$ over instances generated from different benchmarks.}
	\label{fig:ApRobuts}
\end{figure}




\begin{figure}[!t]
	\centering
	\begin{subfigure}[b]{.49\columnwidth}
	\includegraphics[width = \columnwidth]{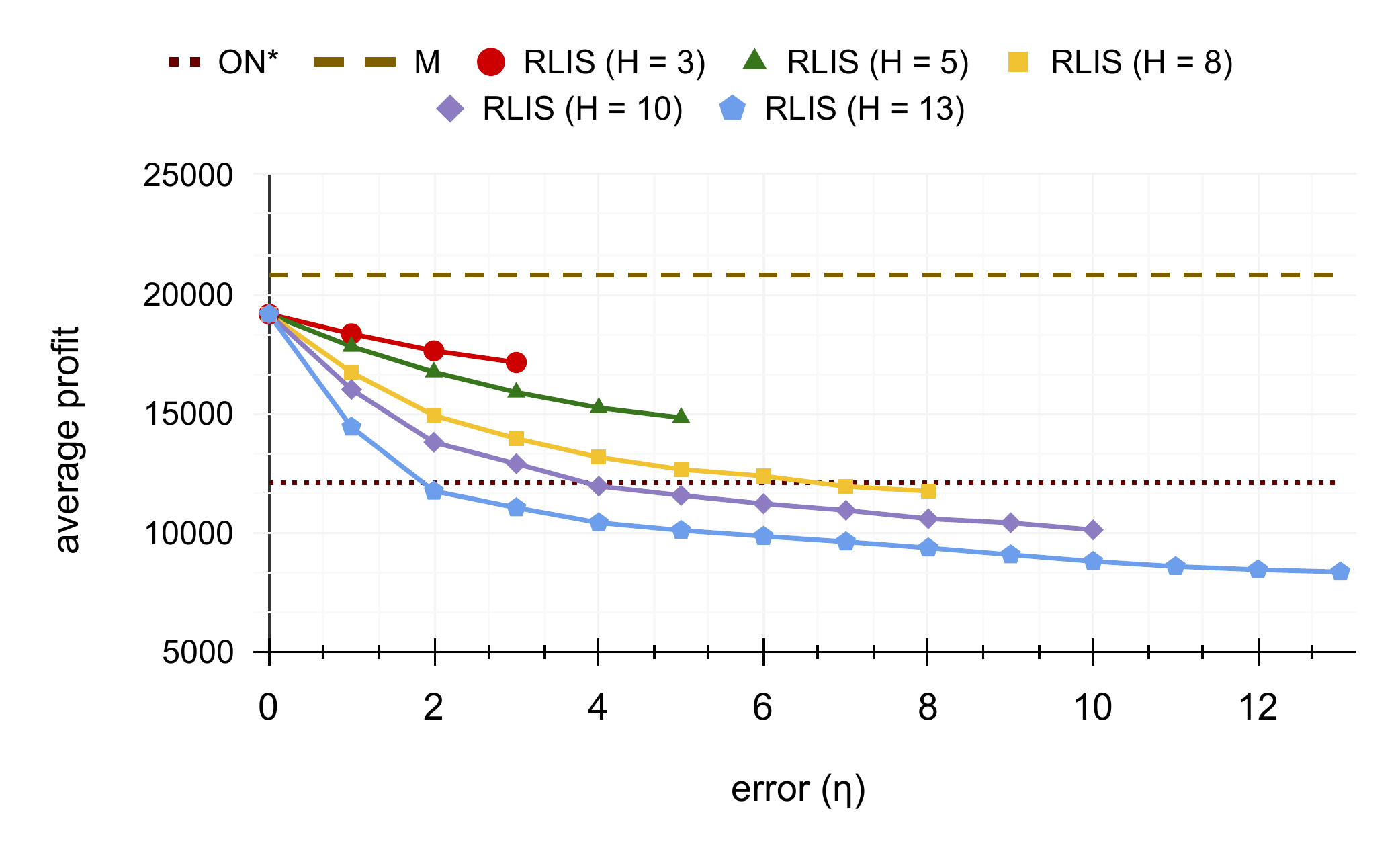}
	\caption{average profit for the \bitcoin benchmark}
	\end{subfigure} \hfill
	\begin{subfigure}[b]{.49\columnwidth}
		\includegraphics[width = \columnwidth]{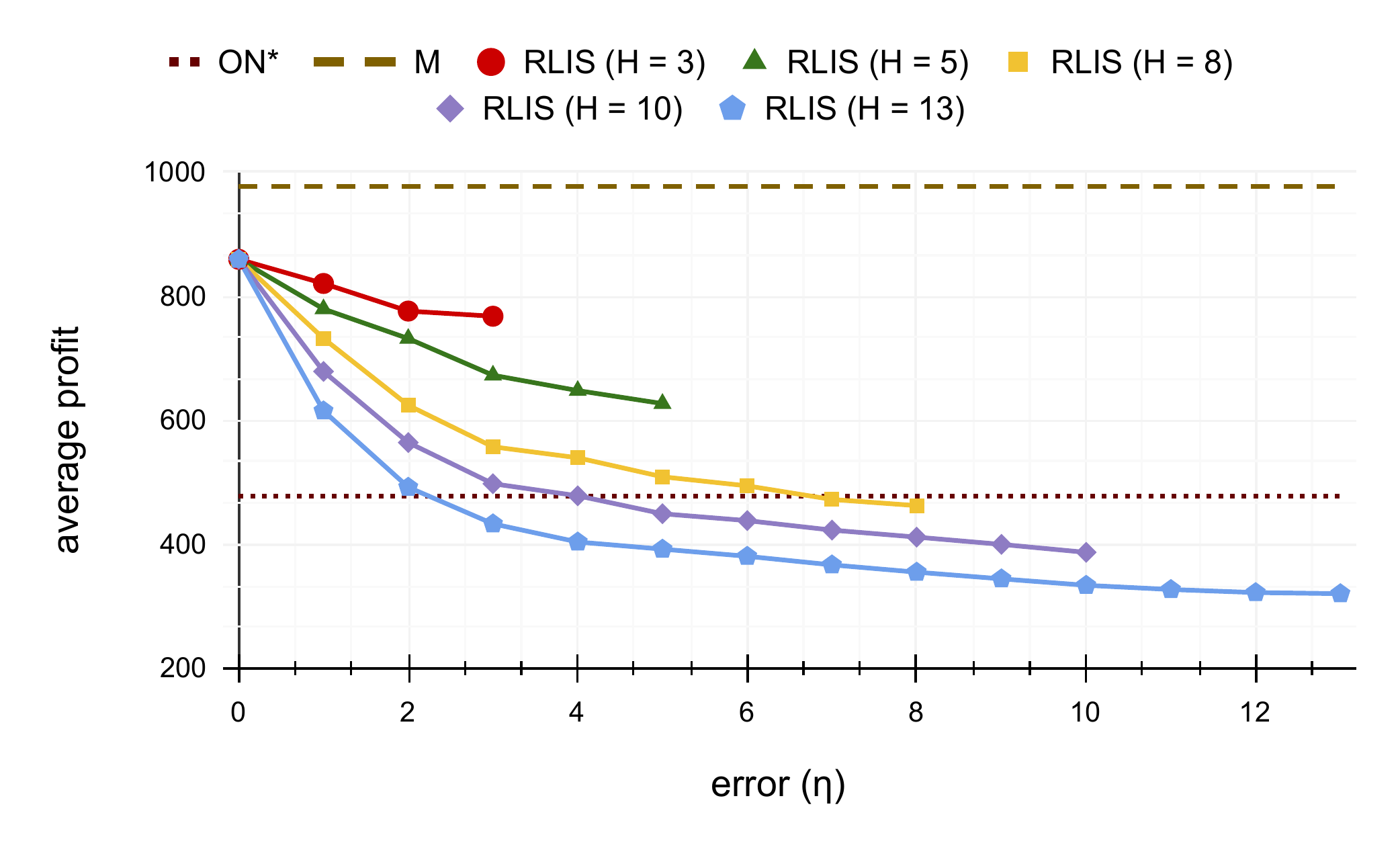}
		\caption{average profit for the \ether benchmark}
	\end{subfigure}
	\hfill
	\begin{subfigure}[b]{.49\columnwidth}
		\centering
		\includegraphics[width = \columnwidth]{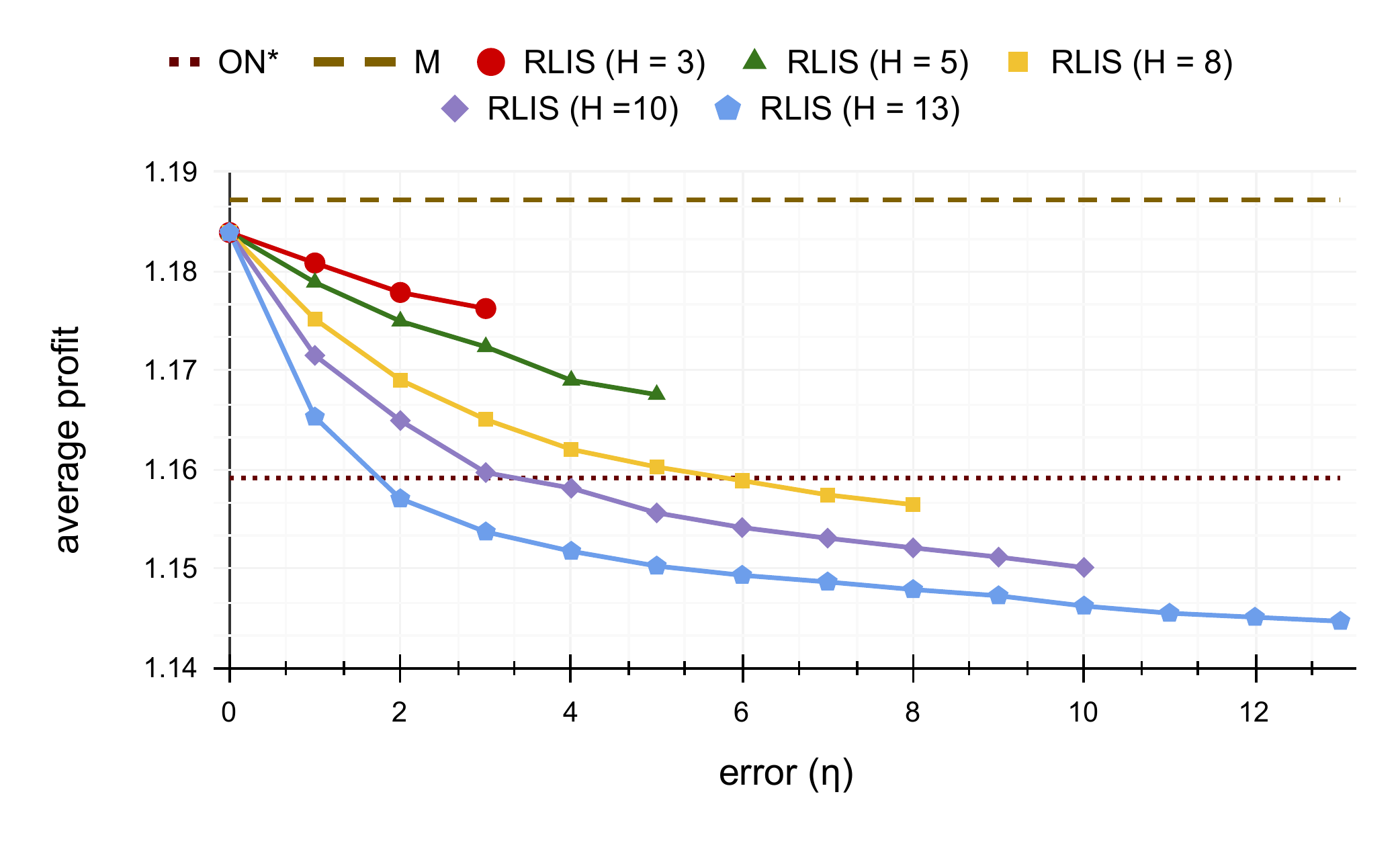}
		\caption{average profit for the \euro benchmark}
	\end{subfigure}
	\hfill
	\begin{subfigure}[b]{.49\columnwidth}
		\centering
		\includegraphics[width = \columnwidth]{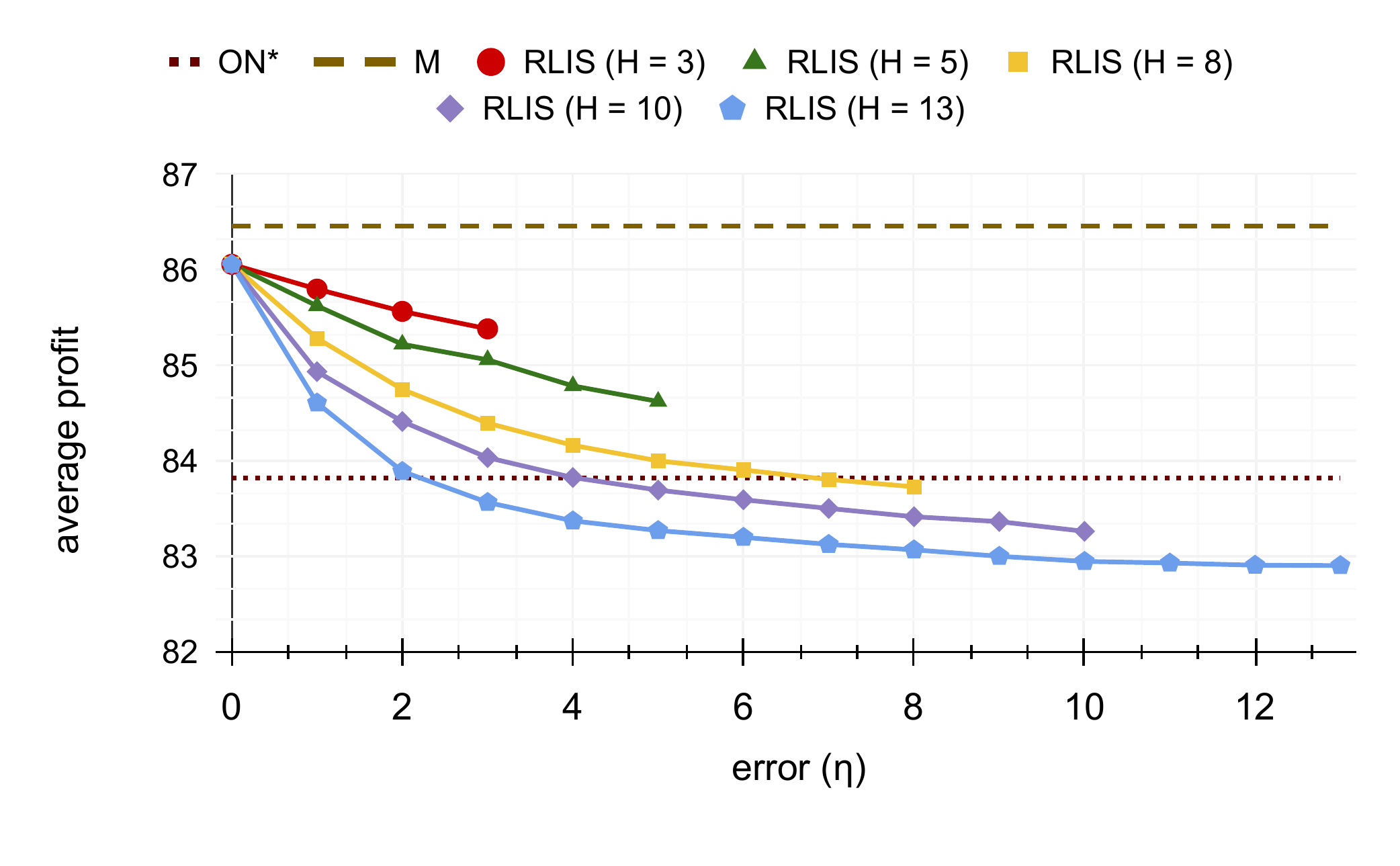}
		\caption{average profit for the \jpy benchmark}
	\end{subfigure}
	\caption{The average profit of \RLIS with different values of $H$ over instances generated from different benchmarks.}
	\label{fig:bitcoinlinear0}
\end{figure}

\subsection{Query-based algorithms}
In our experiments, we set the number of queries to $n=25$. 
We test \RLIS and \RBIS with $H$ taken from $\{ 3, 5, 8, 10, 13 \}$,
and for all values of $\eta \in[0,H]$.  
Let $Alg$ denote any of our algorithms (\RLIS or \RBIS for a certain value of $H$). 
For each instance $I_x$ from our benchmarks and each selected value of $\eta_0$, the following process is repeated 1000 times for $Alg$. First, the (correct) responses to the 25 queries asked by $Alg$ are generated; then out of these 25 responses, $\eta_0$ of them are selected uniformly at random, and flipped. This is the prediction $P$ that is given to $Alg$; we run $Alg$ with this prediction, and record its profit. After running 1000 tests, the average value of the reported profits is recorded as the average profit of $Alg$ for $I_x$, for a value of error equal to $\eta_0$. 

Figures~\ref{fig:bitcoinlinear0} and~\ref{fig:bitcoinbinary0} depict the average profit (as a function of $\eta$) for \RLIS and \RBIS, respectively. 
Since this is a non-oblivious setting, the profit is only reported for values of $\eta \leq H$. The results are consistent over all benchmarks. We observe that both algorithms attain profit significantly better than \ONstar for reasonable values of error, and their profit degrades gently with the error. In particular, \RBIS with $H \in \{3, 5\}$ accrues an optimal profit. For a fixed value of $\eta$, smaller values of $H$ yield to better profit for both algorithms. This is consistent with Theorems~\ref{th:main:linear} and~\ref{th:main:binary}, which bound the competitive ratios as an increasing function of $H$. We also observe that \RBIS performs better than \RLIS, which is again consistent with Theorems~\ref{th:main:linear} and~\ref{th:main:binary}. We also observe that even if $H$ is relatively large (e.g., $H=8$), 
\RBIS results in better profit in comparison to \ONstar.



\begin{figure}[!t]
	\centering
	\begin{subfigure}[b]{.49\columnwidth}
	\includegraphics[width = \columnwidth]{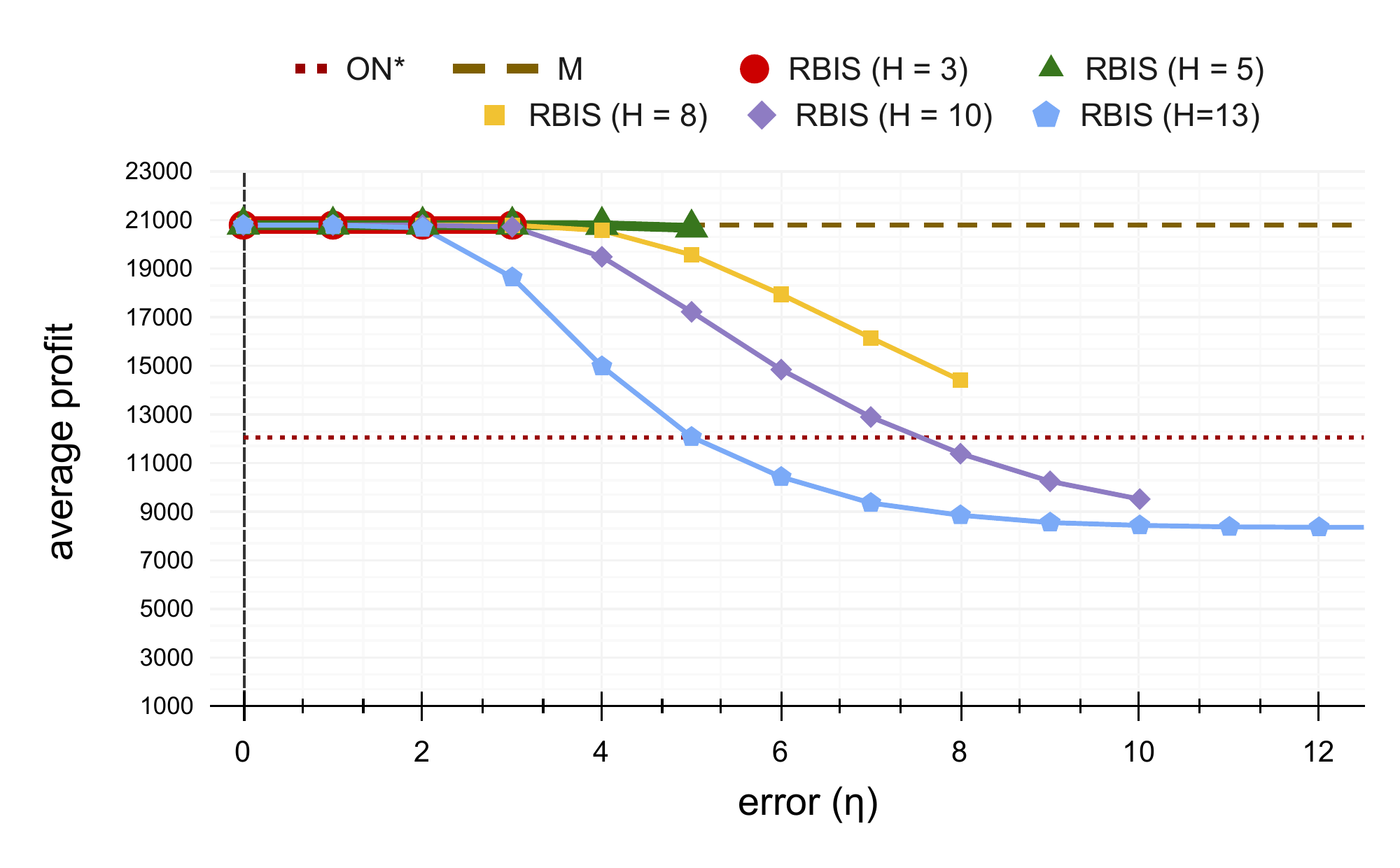}
	\caption{average profit for the \bitcoin benchmark}
	\end{subfigure}
	\begin{subfigure}[b]{.49\columnwidth}
		\includegraphics[width = \columnwidth]{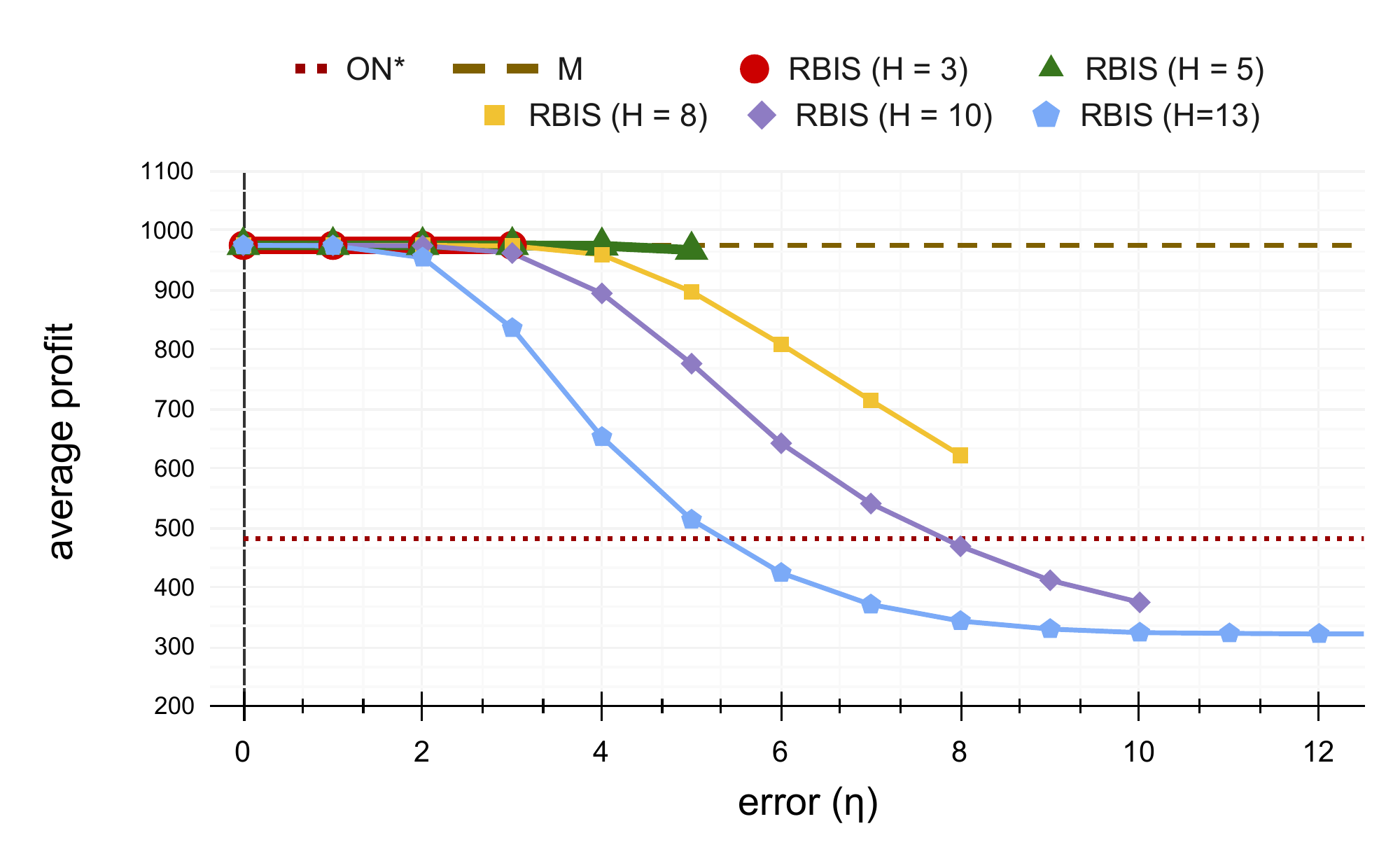}
		\caption{average profit for the \ether benchmark}
	\end{subfigure}
	\hfill
	\begin{subfigure}[b]{.49\columnwidth}
		\centering
		\includegraphics[width = \columnwidth]{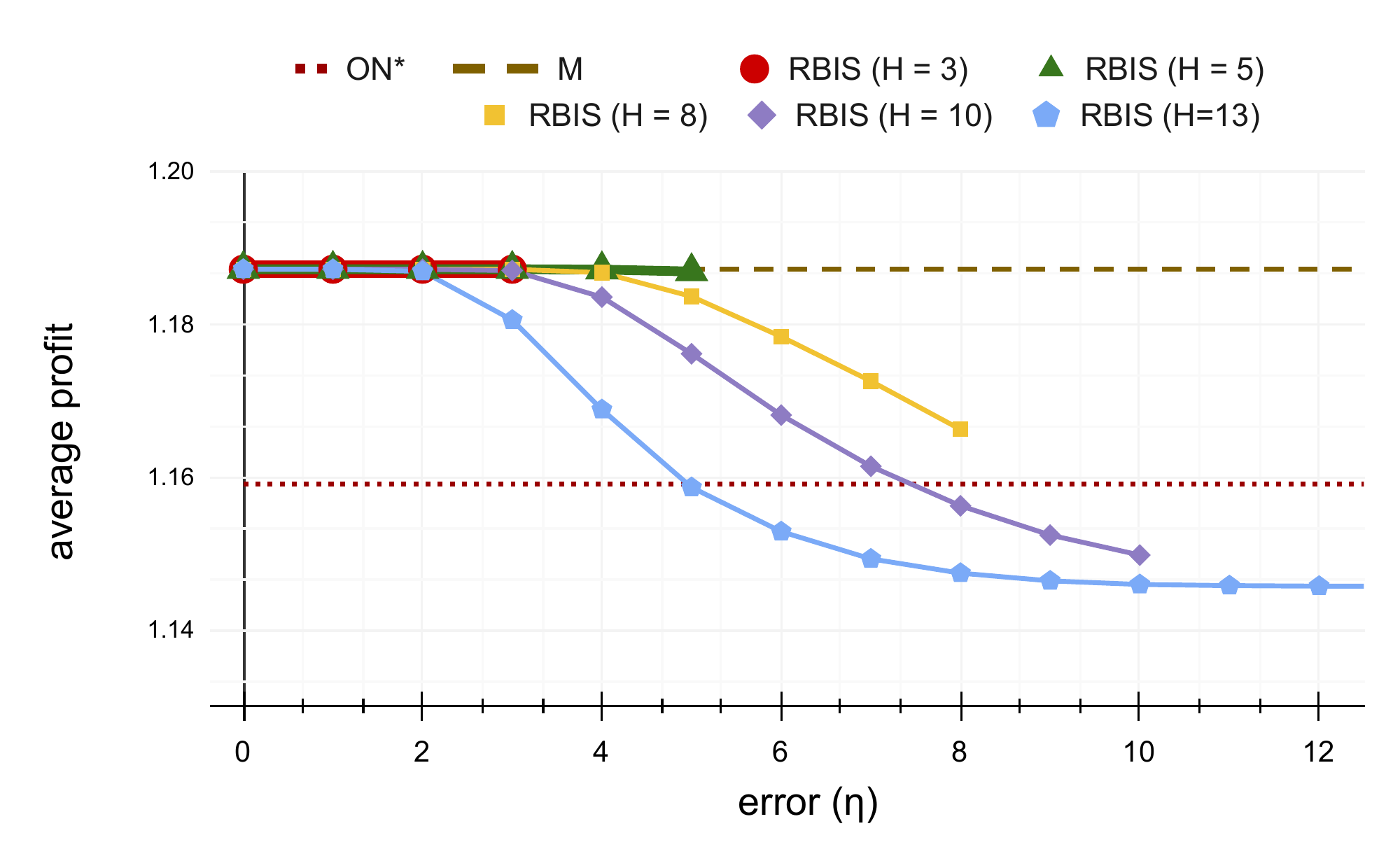}
		\caption{average profit for the \euro benchmark}
	\end{subfigure}
	\hfill
	\begin{subfigure}[b]{.49\columnwidth}
		\centering
		\includegraphics[width = \columnwidth]{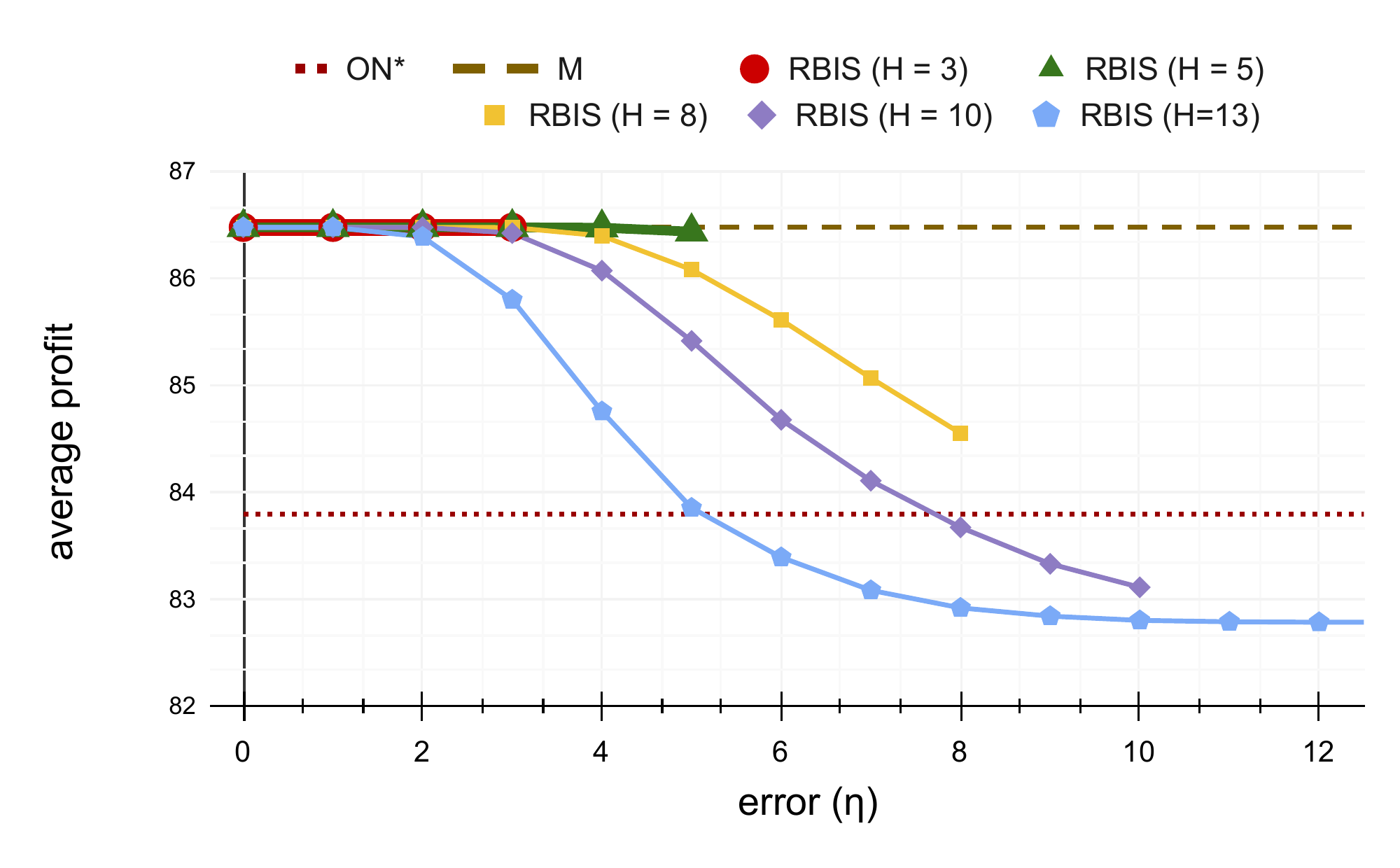}
		\caption{average profit for the \jpy benchmark}
	\end{subfigure}
	\caption{The average profit of \RBIS with different values of $H$ over instances generated from different benchmarks.}
	\label{fig:bitcoinbinary0}
\end{figure}

\section{Conclusion}
\label{sec:conclusion}

We gave the first theoretical study, with supporting experimental evaluation over real data, of a fundamental problem in online decision making, and in a learning-augmented setting. Despite the simplicity of the problem in its standard version, the learning-augmented setting is quite complex and poses several challenges. Future work should expand the ideas in this work to generalizations of online search such as 
one-way trading and online portfolio selection. 

Our robust binary search algorithm can be useful in other query-based optimization settings, with or without predictions, since it addresses a broad setting: select a ``good'' candidate, using noisy queries, while maximizing the size of the candidate space (exponential in the number of queries). 

\bibliographystyle{plain}
\bibliography{aaai22,refs}

\end{document}